%% file: main.tex
\documentclass{IEEEtaes}

\usepackage{algorithm}
\usepackage{algpseudocode}
\usepackage{amsmath}
\usepackage{graphicx}

\usepackage{amssymb}
\usepackage{color,array,amsthm}
\usepackage{graphicx}
\usepackage{comment}
\usepackage[acronym,shortcuts]{glossaries}

\usepackage{subcaption}
\usepackage{tikz}
\usetikzlibrary{arrows,automata, positioning}
\usetikzlibrary{calc, chains,
                fit,
                positioning,
                shapes}
\usetikzlibrary{patterns} 
\usepackage{pgfplots}
\pgfplotsset{compat=1.17}

\newtheorem{prop}{Proposition}
\newtheorem{cor}{Corollary}
\newtheorem{theorem}{Theorem}

\newtheorem{Rem}{Remark}
\newtheorem{Exa}{Example}
\newacronym{AWGN}{AWGN}{Additive White Gaussian Noise}
\newacronym{BCH}{BCH}{Bose-Chaudhuri-Hocquenghem}
\newacronym{BPSK}{BPSK}{Binary Phase-Shift Keying}
\newacronym{CCSDS}{CCSDS}{Consultative Committee for Space Data Systems}
\newacronym{CER}{CER}{Codeword Error Rate}
\newacronym{CLTU}{CLTU}{Communications Link Transmission Unit}
\newacronym{CSP}{CSP}{Closest String Problem}
\newacronym{FSP}{FSP}{Farthest String Problem}
\newacronym{LDPC}{LDPC}{Low-Density Parity-Check}
\newacronym{LFSR}{LFSR}{Linear-Feedback Shift-Register}
\newacronym{LLR-SPA}{LLR-SPA}{Log-Likelihood Ratio Sum-Product Algorithm}
\newacronym{LRT}{LRT}{Likelihood Ratio Test}
\newacronym{ML}{ML}{Maximum Likelihood}
\newacronym{MSA}{MSA}{Min-Sum Algorithm}
\newacronym{NMSA}{NMSA}{Normalized Min-Sum Algorithm}
\newacronym{QC}{QC}{Quasi-Cyclic}
\newacronym{ROC}{ROC}{Receiver Operating Characteristic}
\newacronym{DTS}{RTS}{Receiver Tail Sequence}
\newacronym{S-LRT}{S-LRT}{Simplified Likelihood Ratio Test}
\newacronym{SNR}{SNR}{Signal-to-Noise Ratio}
\newacronym{SPA}{SPA}{Sum-Product Algorithm}
\newacronym{TC}{TC}{TeleCommand}
\newacronym{TS}{TS}{Tail Sequence}

\newcommand{\F}[1]{\textcolor{red}{#1}}

\begin{document}

\title{Design and Analysis of the Tail Sequence for Short LDPC-Coded Space Communications}

\author{MASSIMO BATTAGLIONI}
\member{Member, IEEE}
\affil{Università Politecnica delle Marche, 60131 Ancona, Italy, \& Consorzio Nazionale Interuniversitario per le Telecomunicazioni (CNIT),
43124 Parma, Italy} 

\author{KENNETH ANDREWS}
\affil{Jet Propulsion Laboratory, California Institute of Technology, Pasadena, CA, USA} 

\author{REBECCA GIULIANI}
\member{Student Member, IEEE}
\affil{Università Politecnica delle Marche, 60131 Ancona, Italy, \& Consorzio Nazionale Interuniversitario per le Telecomunicazioni (CNIT),
43124 Parma, Italy}

 \author{FABRIZIO MARINELLI}
\affil{Università Politecnica delle Marche, 60131 Ancona, Italy}

\author{FRANCO CHIARALUCE}
\member{Senior Member, IEEE}
\affil{Università Politecnica delle Marche, 60131 Ancona, Italy, \& Consorzio Nazionale Interuniversitario per le Telecomunicazioni (CNIT),
43124 Parma, Italy}

\author{MARCO BALDI}
\member{Senior Member, IEEE}
\affil{Università Politecnica delle Marche, 60131 Ancona, Italy, \& Consorzio Nazionale Interuniversitario per le Telecomunicazioni (CNIT),
43124 Parma, Italy}

\receiveddate{This article has been accepted in IEEE Transactions on Aerospace and Electronic Systems. To cite it, please use: \emph{M. Battaglioni, K. Andrews, R. Giuliani, F. Marinelli, F. Chiaraluce and M. Baldi, ``Design and Analysis of the Tail Sequence for Short LDPC-Coded Space Communications,'' in IEEE Transactions on Aerospace and Electronic Systems, doi: 10.1109/TAES.2025.3599498.} \\
``The work of Massimo Battaglioni, Fabrizio Marinelli, Marco Baldi, and Franco Chiaraluce was partially supported by the European Union - Next Generation EU under the Italian National Recovery and Resilience Plan (NRRP), Mission 4, Component 2, Investment 1.3, CUP J33C22002880001, partnership on “Telecommunications of the Future” (PE00000001 - program “RESTART”).'' This work was carried out in part at the Jet Propulsion Laboratory, California Institute of Technology, under contract with the National Aeronautics and Space Administration.}

\corresp{{\itshape (Corresponding author: M. Battaglioni)}.}

\authoraddress{Massimo Battaglioni, Rebecca Giuliani, Franco Chiaraluce, and Marco Baldi are with the Department of Information Engineering at Università Politecnica delle Marche, 60131 Ancona, Italy, \& Consorzio Nazionale Interuniversitario per le Telecomunicazioni (CNIT), 43124 Parma, Italy.
Fabrizio Marinelli is with the Department of Information Engineering at Università Politecnica delle Marche, 60131 Ancona, Italy. 
(e-mail: \href{mailto:m.battaglioni@univpm.it,r.giuliani@pm.univpm.it,fabrizio.marinelli@univpm.it,f.chiaraluce@univpm.it,m.baldi@univpm.it}{m.battaglioni@univpm.it, r.giuliani@pm.univpm.it, fabrizio.marinelli@univpm.it, f.chiaraluce@univpm.it, m.baldi@univpm.it}). Kenneth Andrews is with the Jet Propulsion Laboratory, California Institute of Technology, Pasadena, CA (e-mail: \href{mailto:kenneth.s.andrews@jpl.nasa.gov}{kenneth.s.andrews@jpl.nasa.gov).}}

\markboth{M. BATTAGLIONI ET AL.}{Tail Sequence for Short Communications}
\maketitle

\begin{abstract}
According to some standards for satellite communications, the transmitted stream is divided into transmission units with variable length, for which detecting the termination is  particularly relevant. This is the case of space TeleCommands (TCs), where coded data are usually preceded by a start sequence, and optionally followed by a tail sequence, forming the Communication Link Transmission Unit (CLTU).
Regarding the choice of schemes for error correction, the Consultative Committee for Space Data Systems recommendations for TC  synchronization and coding suggests to use, among others, two Low-Density Parity-Check (LDPC) codes: one (relatively) long and one short.
Adopting the long LDPC code eliminates the need for a tail sequence, as the LDPC decoder always fails when overrunning the end of the CLTU, thus causing the decoding and detection process to stop.
This, however, is not true when the short LDPC code is adopted, since its decoding might converge on a codeword even when the decoder input is not a noisy codeword.
This makes it necessary to use a tail sequence that causes the decoder to fail regardless of its input.
In this paper, we study the features required for such a sequence and propose some methods for its design.
Our numerical results, obtained considering various detection approaches for the tail sequence, show that the overall TC rejection probability improves significantly when the proposed tail sequence is employed. Our simulations also show that, for moderate values of the \ac{SNR}, with a properly designed tail sequence it is possible to obtain the same  performance  in terms of TC rejection probability using decoder-based detection and likelihood ratio test-based detection, with the former approach being less complex than the latter.
\end{abstract}


\section{INTRODUCTION}

Deep space communications are gaining increasing attention as enablers of future space exploration missions. Given the extreme distances and challenging conditions, the success of these missions critically depends on the robustness and reliability of communication systems. 
The \ac{CCSDS} has provided comprehensive guidelines to delineate the architecture of satellite data streams organized in \acp{CLTU}, integrating advanced error correction coding \cite{bluebook}.
\Acp{TC}, which are essential for the operation and management of a spacecraft, rely heavily on the integrity of each received \ac{CLTU}. 
Regarding error correction coding techniques, these transmissions traditionally employ a modified \ac{BCH} code  (whose family was introduced in \cite{hocquenghem1959codes,bose1960class}), to which two state-of-the-art \ac{LDPC} codes were later added. The standard accounts for two \ac{LDPC} code variants: a shorter one, with block length of $128$ bits and dimension of $64$ bits, and a longer one, with block length of $512$ bits and dimension of $256$ bits.

\Acp{CLTU} are structured to include a start sequence, a series of encoded data codewords, and, depending on the coding scheme, an optional \ac{TS}. The inclusion of a termination mechanism becomes necessary when the length of the coded data is variable or not known a priori. A lightweight termination strategy involves appending a deliberately undecodable sequence as the \ac{TS}, relying on the decoder's inability to interpret it as a noisy version of a valid codeword and decode it into any codeword. The resulting decoding failure is then used to mark the end of the \ac{CLTU}. This approach is particularly effective for \ac{BCH} codes decoded with bounded-distance decoders, which are incomplete decoders and reliably reject any pattern beyond their correction capability. 
The situation is more complex when \ac{LDPC} codes are used in combination with iterative belief propagation decoders. These decoders are also incomplete and can thus fail to converge to a valid codeword, returning a decoding failure. This property can, in principle, be exploited to detect the end of the \ac{CLTU} without the need for a dedicated detector. In practice, this detecting strategy is particularly effective when  \ac{LDPC} codes with moderate to large block lengths are employed, as the decoder reliably fails on typical idle patterns (such as alternating bits) that follow the data. We have verified this behavior for the $(512,256)$ \ac{LDPC} code specified in~\cite{bluebook}, as well as for the $(256,128)$ \ac{LDPC} code described in~\cite{ccsds_orange}. 
 However, for short \ac{LDPC} codes (like the one with block length of $128$ bits), the limited minimum distance may prevent the decoder from failing on such sequences, unless the termination string is carefully designed. In fact, smaller minimum distance translates into a reduced ability to identify sequences that consistently lead to decoding failures across different channel conditions. In such cases, the use of a well-designed \ac{TS} becomes essential to avoid decoder overrun and to ensure robust detection.

\subsection{Our Contribution}

The design and analysis of \acp{TS} in space data systems employing short \ac{LDPC} codes is relatively unexplored in the literature. Some insights have been discussed in \ac{CCSDS} technical presentations \cite{bertinelli_2014, bertinelli_2015, slide_ken}, which contributed to the initial formulation of the \ac{TS} during the standardization process. A more detailed analysis of \ac{TS} performance was then presented in \cite{Baldi2016}, where the effectiveness of the standardized \ac{TS} was assessed in the context of \ac{BCH} coding. However, such a study did not extend to \ac{LDPC}-coded systems, leaving an open question regarding the \ac{TS} behavior under modern error correction schemes. Given the increasing adoption of short \ac{LDPC} codes in space communications, a dedicated investigation is necessary to assess their impact on \ac{TS} performance. 

This paper aims to fill such gap by proposing novel design methodologies for \acp{TS} to improve the reliability of short \ac{LDPC}-coded \Acp{TC}. We  perform  a theoretical analysis and comprehensive numerical simulations. In particular, we examine various detection mechanisms, such as correlators, \acp{LRT} \cite{massey1972optimum, Chiani2006, Pfletschinger2015}, and decoder-based strategies to assess the impact of different \ac{TS} detection methods on the overall system performance. The \acp{TS} obtained with the proposed method are effective even when different detection techniques are applied. This effectiveness derives from the observation that the proposed \acp{TS} can accommodate all distinct requirements such techniques have. For instance, tail detectors should operate on sequences that yield relatively small probabilities of false alarm (when operating on noisy codewords) and missed detection (when operating on noisy \acp{TS}) according to some metric, whereas decoder-based detectors rely on \acp{TS} that induce an ``uncorrectable'' error pattern ensuring, with high probability, that the decoder does not mistakenly recognize them as valid codewords. These requirements are not necessarily conflicting; rather, as mentioned, they can be simultaneously satisfied by appropriately designed \acp{TS}. 

The \ac{TC} rejection probability, that is, the probability that the \ac{TC} is not recognized by the satellite, serves as a quantifiable measure of the rate at which \acp{CLTU} are erroneously dismissed. In \cite{GiulianiAccess2025}, the performance of the system defined by the current standard is studied, but the \ac{TS} design is not faced. In this paper, by analyzing the performance implications of \ac{TS} design and detection, we aim to provide insights that significantly improve the robustness of \ac{TC} systems.

In fact, we show that it is possible to achieve the same performance in terms of \ac{TC} rejection rate as \ac{LRT}-based detectors by relying exclusively on the decoder already used for recovering noisy encoded data, thereby reducing both latency and receiver complexity. Unlike the \ac{LRT}-based approach, this method eliminates the need to switch between separate decoding and detection systems, making on-board \ac{TC} operations smoother. 
In fact, in the \ac{LRT}-based scheme, after detection of a preamble by the ``start'' detector, the receiver must alternate between the ``tail'' detector and the decoder on the same window of $n$ received symbols, invoking decoding only when tail detection fails. Moreover, the ``tail'' and ``start'' detectors are structurally different: the former jumps blockwise, while the latter operates sequentially, introducing further complexity in implementation and control. By contrast, the approach we consider unifies the detection and decoding steps through a single mechanism, leveraging the LDPC decoder alone. The only added cost is a one-time, off-line optimization of the \ac{TS}, which remains computationally affordable in the short-blocklength regime considered, especially since it must not be performed on board.

\subsection{Related works}

We observe that a strategy for designing an uncorrectable error pattern corresponds to finding a solution to the \ac{FSP}, i.e., a word that is as distant as possible from all codewords according to a chosen metric. The farthest string is the one (possibly not unique) that realizes the so-called covering radius. 
In fact, the covering radius of a set $P$ in a metric space is the smallest radius $r$ such that every point in the space lies within distance $r$ from some point in $P$. In a Hamming space, the maximum Hamming distance to the nearest codeword gives the covering radius.
The \ac{FSP} is also connected to the \ac{CSP}, since it is well known that the covering radius in an $n$-dimensional Hamming space corresponds to $n$ minus the Chebyshev radius of the smallest ball that contains all the codewords, 
whose center is referred to as the Chebyshev center, or equivalently, the closest string.  
Computing the farthest and the closest string from a set of $m$ $b$-ary strings ($b\ge 2$) is generally NP-hard \cite{Lanctot2003}. The covering radius problem has been shown to be NP-hard both for an arbitrary code, explicitly represented by  $m$ codewords \cite{Frances1997}, and for a linear code specified by its parity-check matrix \cite{McLoughlin1984,Guruswami2005}. In particular, McLoughlin showed that even verifying an upper bound for the covering radius is $\Pi_2^p$-complete, and that computing the covering radius is both $\Pi_2^p$-hard and $\Sigma_2^p$-hard  \cite{McLoughlin1984}.\\ 

A polynomial-time $(1 - \epsilon)$-approximation scheme (PTAS) with a time complexity of $m^{O\left(\frac{1}{\epsilon^2}\right)}$ for the related Chebyshev radius problem was recently proposed in \cite{Mazumdar2013}. Since the algorithm is polynomial-time only for linear codes whose dimension $k$ is logarithmic in the block length $n$, which generally are of little practical interest, it would be interesting to study the approximability of \ac{FSP} on \ac{LDPC} codes, which have sparser parity-check matrices than general linear codes but still have exponentially many codewords, and are of much greater practical relevance.

The solutions of the \ac{CSP} and \ac{FSP} have been widely investigated in the literature (see, among many others, \cite{Li2002,Meneses2004, Dinu2019}). However, most of these approaches do not take into account the fact that the strings under consideration may be codewords. This limits the applicability of the proposed techniques in real-world satellite communication scenarios and highlights the necessity of a search algorithm that leverages the code structure.

\subsection{Paper outline}

In Section \ref{sec:preli} we establish the necessary notation and preliminaries. In Section \ref{sec:tsdes} we deal with the \ac{TS} design. Section \ref{sec:numres} contains performance assessments, with particular emphasis on the \ac{TC} rejection probability.  
Finally, Section \ref{sec:concl} provides concluding remarks.

\section{Preliminaries}\label{sec:preli}

In this section, we introduce the notation used throughout the paper and we describe the considered communication system.

\subsection{Notation}\label{subsec:not}

We use the notation $[a,b]$ to represent the set of integers between $a$ and $b$, endpoints included. 
To denote vectors and matrices, we use bold lowercase and uppercase letters, respectively. $\mathbf{I}_M$ is the $M \times M$ identity matrix, $\mathbf{\Phi}^i$ is the $i$-th right circular shift of $\mathbf{I}_M$, where ${0\leq i \leq M-1}$,  $\mathbf{0}_M$ is the $M \times M$ zero matrix, and $\mathbf{0}_{1\times M}$ is an all-zero vector of length $M$. The $\oplus$ operator indicates modulo-2 addition. $\mathbb{F}_2$ is the Galois field of order $2$. Given a set $S$, we denote its cardinality as $|S|$.  The support of a vector $\mathbf{a}$ of length $n$ is defined as
$\mathrm{Supp}(\mathbf{a})= \{j\in [0 , n-1]\hspace{2mm}|\hspace{2mm}a_j\neq 0\}$. The Hamming weight of a vector $\mathbf{a}$ coincides with its $\ell_1$-norm, and we denote it by  $||\mathbf{a}||=|\mathrm{Supp}(\mathbf{a})|$. We define $\mathbf{e}_i$ as the vector such that $\mathrm{Supp}(\mathbf{e}_i)=\{i\}$. The function $d_{\mathrm{H}}(\cdot,\cdot)$ returns the Hamming distance between two vectors, defined as the number of entries where they differ.    Given two vectors $\mathbf{a}$ and $\mathbf{b}$ with the same length, we say that their overlap is $O(\mathbf{a},\mathbf{b})=|\mathrm{Supp}(\mathbf{a}) \cap \mathrm{Supp}(\mathbf{b})|$.

For the sake of clarity, let us also provide the necessary background on error correcting codes, focusing on binary linear codes.
A binary linear code $\mathcal{C}$ with dimension $k$ and block length $n$, denoted as $(n,k)$ code, is a linear $k$-dimensional subspace of  $\mathbb{F}_2^n$. The code rate  is $R=\frac{k}{n}$. A linear code $\mathcal{C}$ can be represented as the kernel of a parity-check matrix $\mathbf{H}$; that is,
\[
\mathcal{C}=\{\mathbf{c}\,|\,\mathbf{c}\mathbf{H}^\top=\mathbf{0}\},
\]
where $\mathbf{c}$ is a codeword, $^\top$ denotes transposition, and $\mathbf{0}$ is the all-zero vector. Equivalently, the same code can be represented through a generator matrix $\mathbf{G}$, whose rows form a basis for $\mathcal{C}$, as follows
\[
\mathcal{C}=\{\mathbf{c}\,|\,\mathbf{c}=\mathbf{u}\mathbf{G},\,\mathbf{u}\in\mathbb{F}_2^k\},
\]
where $\mathbf{u}$ is an information sequence.
The generator matrix is in \textit{standard form} if it can be written as  
\[
\mathbf{G} = [\mathbf{I}_{k} \mid \mathbf{P}]
\]
or as a column-permuted equivalent, for some $\mathbf{P}$.

We denote the number of codewords of Hamming weight $w$ as $A(w)$, which represents the weight enumerator function, also called distance distribution. For linear codes, the minimum Hamming distance of the code, $d_{\min}$, is the smallest non-null value of $w$  such that $A(w)>0$. In a linear code, all codewords have identical Hamming distance properties, that is, $A(w)$ also represents the number of codewords at Hamming distance $w$ from any fixed codeword. 

\ac{LDPC} codes are characterized by parity-check matrices with a relatively small number of non-zero entries compared to the number of zeros. A code is said \ac{QC} if any of its codewords can be split into $n_0$ blocks of $M$ bits each, and applying a consistent right (or left) cyclic shift to each block results in another valid codeword. We call this type of shift \emph{quasi-cyclic shift}.

\subsection{Communication system}

The communication system we consider is that described in \cite{bluebook}, and shown in Fig. \ref{fig:blockdiag}. According to the standard, the operations summarized in this scheme should be performed when using \ac{LDPC}-coded \ac{CCSDS} compliant transmissions.

\begin{figure}
    \centering
\resizebox{0.95\columnwidth}{!}{\input{Figs/block}}  
    \caption{Communication system's blocks diagram}
    \label{fig:blockdiag}
\end{figure}
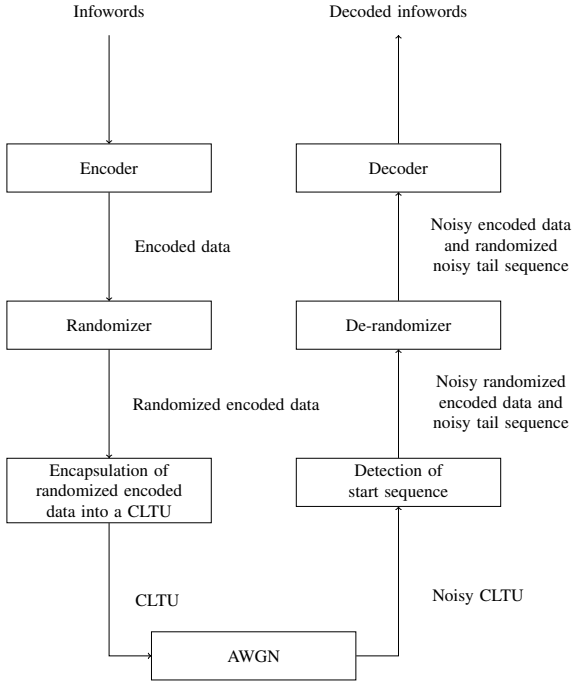 

The $(128,64)$ \ac{LDPC} standard code is specified by an $r\times n$ parity-check matrix $\mathbf{H}_{64 \times 128}$, where $r=n-k=64$ and $n=128$. This matrix is constructed from $M\times M$ submatrices having size $M=k/4=n/8=16$, as shown in Fig. \ref{fig:pcm}. According to Fig. \ref{fig:blockdiag}, encoded data must be randomized before encapsulation. The recommended randomizer  exploits a \ac{LFSR}, represented by the polynomial $x^8+x^6+x^4+x^3+x^2+x+1$. It produces a pseudo-random sequence of period $255$, to be summed mod $2$ to the input.

    \begin{figure*}
    \centering
    $\mathbf{H}_{64 \times 128} = \begin{bmatrix}
    \mathbf{I}_M \oplus \mathbf{\Phi}^7 & \mathbf{\Phi}^2 & \mathbf{\Phi}^{14} & \mathbf{\Phi}^6 & \mathbf{0}_M & \mathbf{\Phi}^0 & \mathbf{\Phi}^{13} & \mathbf{I}_M \\
    \mathbf{\Phi}^6 & \mathbf{I}_M \oplus \mathbf{\Phi}^{15} & \mathbf{\Phi}^0 & \mathbf{\Phi}^1 & \mathbf{I}_M & \mathbf{0}_M & \mathbf{\Phi}^0 & \mathbf{\Phi}^7 \\
    \mathbf{\Phi}^4 & \mathbf{\Phi}^1 & \mathbf{I}_M \oplus \mathbf{\Phi}^{15} & \mathbf{\Phi}^{14} & \mathbf{\Phi}^{11} & \mathbf{I}_M & \mathbf{0}_M & \mathbf{\Phi}^3 \\
    \mathbf{\Phi}^0 & \mathbf{\Phi}^1 & \mathbf{\Phi}^9 & \mathbf{I}_M \oplus \mathbf{\Phi}^{13} & \mathbf{\Phi}^{14} & \mathbf{\Phi}^1 & \mathbf{I}_M & \mathbf{0}_M \\
\end{bmatrix} $
    \caption{Parity-check matrix of the short $(128,64)$ LDPC code}
    \label{fig:pcm}
\end{figure*}

 Then, a typical stream produced by the ground station contains \acp{CLTU}, each followed by a high-transition density sequence (called idle sequence) of variable length, which helps the space segment to keep synchronization until a new \ac{CLTU} is transmitted. A \ac{CLTU} is formed by a start sequence, followed by a variable number of randomized codewords (we denote the maximum number of codewords in a \ac{CLTU} as $N$), and (optionally when the $(128,64)$ \ac{LDPC} code is employed) a \ac{TS}. For short \ac{LDPC}-coded transmissions, the length of the start sequence is $64$ bits, whereas the length of the optional \ac{TS} is the same as that of codewords, i.e., $128$ bits.  

In the subsequent sections of this paper, we demonstrate that excluding the optional \ac{TS} from \acp{CLTU}, which requires relying on the fact that the decoder detects the end of \acp{CLTU} by overrunning it, significantly increases the likelihood of their rejection. Therefore, it is important to use a \ac{TS} and define the best strategy to detect it. The options discussed in \cite{bluebook} are: relying on decoder failures, or using a correlator/\ac{LRT}-based detector. The first approach operates on a straightforward principle: if the decoder fails, both the decoding and \ac{TS} detection processes halt, prompting the upper layers to verify the formal validity of the received \ac{CLTU}. According to the second approach, the detection task is transferred to a separate device, e.g., a correlator/\ac{LRT}-based detector. Should the conditions for \ac{TS} detection be satisfied, the decoding process is terminated, leaving the remainder of the process unchanged.
 Since the \ac{LDPC} decoder processes soft symbols, de-randomization consists of inverting, or not, the signs of the received symbols according to the values of the pseudo-random sequence. In \cite{bluebook}, it is recommended to reset the \ac{LFSR} before randomizing (or de-randomizing) each input $128$-bit sequence. Therefore, the randomizing (or de-randomizing) sequence summed mod $2$ to the codewords, to the \ac{TS}, and to the idle pattern is always the same. Clearly, this sequence matches the first $128$ bits generated by the above \ac{LFSR}, which are, in hexadecimal:

\begin{equation}
\mathbf{s}_{\mathrm{LFSR}}=\mathrm{FF399E5A\,68E906F5\,6C892FA1\,315E08C0}. 
   \label{eq:lfsr}
\end{equation}

\begin{Rem}
    It is shown in \cite{GiulianiAccess2025} that if the \ac{TS} proposed in \cite{bluebook} is randomized before encapsulation, performance significantly improves. For the sake of clarity,  in the rest of the paper we loosely call \ac{DTS} the sequence included in the encapsulation phase summed mod 2 to the first $128$ bits generated by the \ac{LFSR}. This is indeed the \ac{TS} from the decoder viewpoint. In other words, if the \ac{TS} appended to coded data in the encapsulation phase is $\mathbf{t}$, then the \ac{DTS} is
\begin{equation}
\mathbf{v}=\mathbf{t}\oplus\mathbf{s}_{\mathrm{LFSR}}.
\label{eq:randomizets}
\end{equation}
    \label{rem:decTS}
\end{Rem}

\section{Tail sequence design} \label{sec:tsdes}

In this section, we deal with \ac{DTS} design. We first assume that a decoder-based detection is employed. Then,  we show how the \acp{DTS} designed for decoder-based detection can be turned into \acp{DTS} which are also appropriate for correlator/\ac{LRT}-based detection.

Our design approach is based on the Hamming metric. Given $\mathbf{v}\in\mathbb{F}^n_2$ and a code $\mathcal{C}$, we define
\[
D(\mathbf{v},S)= \min_{i \in S} d_{\mathrm{H}}(\mathbf{c}^{(i)}, \mathbf{v}),
\]
where $\mathbf{c}^{(i)}$ denotes the $i$-th codeword\footnote{For the sake of simplicity we consider the codewords $\mathbf{c}^{(i)} \in \mathcal{C}$ to be ordered by increasing Hamming weight and, within the same weight, lexicographically.} and $S\subseteq [0,2^k-1]$ contains the indexes of the codewords from which the Hamming distance is computed. 
Then, the optimal \ac{DTS} $\mathbf{v}^*\in\mathbb{F}_2^n$ should maximize the minimum Hamming distance to any codeword in the code, i.e.,
\begin{equation}
   \mathbf{v}^*=\underset{\mathbf{v}\in\mathbb{F}_2^n}{\mathrm{arg\,max}} \,D(\mathbf{v},[0,2^k-1]). 
   \label{eq:fsp}
\end{equation}
Let \begin{equation}
    D_{\mathcal{C}}=D(\mathbf{v^*},[0,2^k-1]).
    \label{eq:DC}
\end{equation}
For practical instances of this problem, such as the one considered in this paper ($k = 64$), evaluating the distance of a given string from all $2^{64}$ codewords is computationally infeasible. To face this challenge, we propose two alternative design methods. The first restricts comparisons to a smaller subset of codewords, selecting $S$ such that $|S| < 2^k$. The second employs a stochastic local search that dynamically selects subsets of low-weight codewords. 

\subsection{Guaranteed minimum Hamming distance from a fixed subset of codewords}\label{subsec:guara}

In this section we aim at finding an \ac{DTS} $\mathbf{v}_G$ for which
\begin{equation}
    D(\mathbf{v}_G,S)= D(\mathbf{v}_G,[0,2^k-1]),
    \label{eq:DD}
\end{equation}
 being $S= \left[0,\sum_{j = d_{\min}}^{w_{\max}} A(j)\right]\subset [0,2^k-1]$ (where \( d_{\min} \) represents the code minimum distance, and \( w_{\max} \) denotes the maximum Hamming weight of the codewords considered in the analysis), and $D(\mathbf{v}_G,[0,2^k-1])$ is as large as possible.  The first terms of the (exact) weight distribution of the standardized $(128,64)$ \ac{LDPC} code are shown in Table \ref{tab:Aws}. From the table, it is apparent that the code has a minimum distance of  $d_{\min}=14$. All these codewords have been found using the algorithm based on information set decoding proposed in \cite{stern1989method}.

\begin{table}[h]
    \centering
    \caption{$(128,64)$ codeword weight distribution}
    \begin{tabular}{cr}
        \hline
        $w$ & $A(w)$ \\
\hline
        0  & 1 \\
        14 & 16 \\
        16 & 492 \\
        18 & 5\,424 \\
        20 & 81\,760 \\
        22 & 1\,127\,152 \\
        24 & 14\,869\,768 \\\hline
    \end{tabular}
    \label{tab:Aws}
\end{table}

\begin{theorem} 
Let \( \mathcal{C} \) be a binary code of length \( n \) with minimum distance \( d_{\min} \). Let \( \mathbf{v}_G \)  be a sequence such that $||\mathbf{v}_G||=\lfloor(w_{\max}+1)/2\rfloor$ and
\begin{equation}
    O(\mathbf{v}_G, \mathbf{c})\leq \Bigg\lfloor
\frac{||\mathbf{c}||+||\mathbf{v}_G||-(w_{\max}+1)/2}{2}\Bigg\rfloor
\label{eq:operl}
\end{equation} for all codewords \( \mathbf{c} \) with weight \( ||\mathbf{c}|| \in [d_{\min},w_{\max}] \). Then, for all codewords \( \mathbf{c}'\in \mathcal{C} \), it holds that
\begin{equation}
    d_{\mathrm{H}}(\mathbf{v}_G, \mathbf{c}') \geq ||\mathbf{v}_G||.
\label{eq:distfromallcods}
\end{equation}
\label{the:over_gen}
\end{theorem}
\begin{IEEEproof}
For the codewords having weight at least twice that of $\mathbf{v}_G$, it follows from \eqref{eq:operl} that their maximum number of overlapping $1$s with $\mathbf{v}_G$ can reach the maximum possible value, that is, $||\mathbf{v}_G||$. Hence, the Hamming distance between $\mathbf{v}_G$ and any of these codewords is lower bounded by $||\mathbf{v}_G||$. For all the other codewords, the thesis follows from \eqref{eq:operl}. In fact, according to the hypotheses, there are at most $\big\lfloor
\frac{||\mathbf{c}||}{2}\big\rfloor$ overlapping $1$s between any of these codewords and $\mathbf{v}_G$. Therefore, their Hamming distance is at least $d_{\mathrm{H}}(\mathbf{v}_G,\mathbf{c})= ||\mathbf{v}_G||+||\mathbf{c}|| - 2 \Big\lfloor
\frac{||\mathbf{c}||}{2}\Big\rfloor \geq ||\mathbf{v}_G||$.
\end{IEEEproof}

\begin{cor}
    Let \( \mathcal{C} \) be a binary code of length \( n \) with minimum distance \( d_{\min} \) and even $A(w)$ for all non-zero values of $w$. Let \( \mathbf{v}_G \)  be a sequence such that $||\mathbf{v}_G||=\lceil(w_{\max}+1)/2\rceil$ and
\begin{equation}
    O(\mathbf{v}_G, \mathbf{c})\leq \Bigg\lfloor
\frac{||\mathbf{c}||+||\mathbf{v}_G||-(w_{\max}+1)/2}{2}\Bigg\rfloor
\label{eq:operleven}
\end{equation} for all codewords \( \mathbf{c} \) with weight \( ||\mathbf{c}|| \in [d_{\min},w_{\max}] \). Then, for all codewords \( \mathbf{c}'\in \mathcal{C} \), it holds that
\begin{equation}
    d_{\mathrm{H}}(\mathbf{v}_G, \mathbf{c}') \geq ||\mathbf{v}_G||.
\label{eq:distfromallcods_even}
\end{equation}
\end{cor}
\begin{IEEEproof}
    The proof is almost identical to that of Theorem \ref{the:over_gen}. However, here we need to consider that codewords with weight $w_{\max}+1$ do not exist. Therefore, it is possible to set $||\mathbf{v}_G||=\lceil(w_{\max}+1)/2\rceil$, and it is still guaranteed that the maximum number of overlappings with codewords of weight larger than $w_{\max}$ is $|| \mathbf{v}_G||$.
\end{IEEEproof}

An analogous result can be found for codes characterized by odd $A(w)$ for all non-zero values of $w$.

The choice of $w_{\max}$ is, clearly, a trade-off between the algorithm complexity and the \ac{TS} performance. 

The search algorithm we employed to find \acp{DTS} at guaranteed minimum Hamming distance from codewords  is described in Algorithm \ref{alg:tail_sequence_searchfixed}.

\begin{algorithm}
\caption{Receiver Tail Sequence Search with Guaranteed Minimum Hamming Distance from Codewords}
\label{alg:tail_sequence_searchfixed}
\begin{algorithmic}[1]
\renewcommand{\algorithmicrequire}{\textbf{Input:}}
\renewcommand{\algorithmicensure}{\textbf{Output:}}
\Require code length $n$, code minimum distance $d_{\min}$, $w_{\max}>d_{\min}$, all codewords of Hamming weight in $[d_{\min},w_{\max}]$,
\Ensure A valid Receiver Tail Sequence $\mathbf{v}_G$ with guaranteed minimum Hamming distance $\lceil(w_{\max}+1)/2\rceil$ from all codewords.
\State \textbf{Initialization:} 
\State $D \gets \lceil(w_{\max}+1)/2\rceil$
\State success $\gets$ false
\Repeat
    \State Generate a random binary pattern $\mathbf{v}$ of length $n$ with weight $D$
    \State valid $\gets$ true
    \For{$\omega \in [d_{\min},w_{\max}]$}
        \State $O_{\omega} \gets  \frac{\omega}{2}$
        \For{$\mathbf{c}$ such that $\mathbf{c}\in\mathcal{C}$ and $\mathrm{wt}_{\mathrm{H}}(\mathbf{c})=\omega$ }
            \If{$O(\mathbf{v},\mathbf{c})>O_{\omega}$}
                \State valid $\gets$ false
                \State \textbf{break}
            \EndIf
        \EndFor
        \If{not valid}
            \State \textbf{break}
        \EndIf
    \EndFor
    \If{valid}
        \State success $\gets$ true
        \State \textbf{break}
    \EndIf
\Until{maximum attempts reached}

\If{not success}
    \State \Return \texttt{failure}
\Else
    \State \Return $\mathbf{v}_G$
\EndIf
\end{algorithmic}
\end{algorithm}

For the standard $(128,64)$ \ac{LDPC} code, if we run Algorithm \ref{alg:tail_sequence_searchfixed} considering $w_{\max}=22$, which yields $||\mathbf{v}_G|| = 12$, we are able to obtain several candidate \acp{DTS}, some of which are shown in Table \ref{tab:tsguara}. However, we anticipate that this design method does not have optimal performance for the specific application we are considering, since the final weight of these \acp{DTS} is too low. For this reason, in the next section we propose a different approach to the \ac{DTS} design.
 
\begin{table}[h]
\centering
\caption{Hexadecimal representations of candidate Receiver \acp{TS} with guaranteed minimum Hamming distance $12$ from codewords}
\begin{tabular}{|c|c|}
\hline
\textbf{Sequence Nr.} & \textbf{Hexadecimal Representation of $\mathbf{v}_G$} \\
\hline
1 & 40002103\,20000001\,00200000\,04040009 \\
2 &  10000080\,03006000\,00600008\,000000C1\\
3 & 00100002\,80402080\,01001002\,00401200 \\
4 & 02080200\,40000000\,08108001\,A0005000 \\
5 & 80280000\,04000000\,03044200\,00010804 \\
\hline
\end{tabular}
\label{tab:tsguara}
\end{table}

\subsection{Stochastic local search}\label{subsec:greedy}
 This section explores the potential for increasing the minimum Hamming distance between the \ac{DTS} and a dynamically selected subset of low-weight codewords, where the subset is not predetermined as in Section \ref{sec:tsdes}-\ref{subsec:guara}, but evolves as the algorithm progresses. 

Let us consider $\mathbb{F}_2^n$ equipped with the Hamming distance function. Given a linear code $\mathcal{C}$, we define the \emph{Voronoi cell} around codeword $\mathbf{c}^{(i)}$, $i\in[0,2^k-1]$, denoted as $\mathcal{V}(\mathbf{c}^{(i)})$, the set of points in $\mathbb{F}_2^n$ whose Hamming distance from $\mathbf{c}^{(i)}$ is smaller than or equal to the Hamming distance from $\mathbf{c}^{(j)}$, where $i\neq j$. Therefore,
\[
\mathcal{V}(\mathbf{c}^{(i)})=\{\mathbf{v}\in\mathbb{F}_2^n \,|\, d_{\mathrm{H}}(\mathbf{v},\mathbf{c}^{(i)})\leq d_{\mathrm{H}}(\mathbf{v},\mathbf{c}^{(j)}), \forall i\neq j\}.
\]
Since the code is linear, all the Voronoi cells have the same shape.

In this section, we aim to identify vectors that lie on the boundaries\footnote{In $\mathbb{F}_2^n$, the notion of a boundary is not as rigorously defined as in Euclidean spaces. Here, we loosely refer to the boundary of a Voronoi cell as the set of vectors at maximal Hamming distance from the codeword that defines the cell.} between Voronoi cells, as potential candidates for \acp{DTS}. The criteria used to assess the quality of these candidates are presented below and are, naturally, based on their distance from the nearest codewords.  We propose to use a stochastic local search which starts from a codeword and looks for vectors in the boundary between Voronoi cells.

The stochastic local search we propose is described in Algorithm \ref{alg:greedy}. The idea is to start from a Hamming weight-1 vector $\mathbf{s}$, and then progressively flip its entries, one by one, as follows. An entry of $\mathbf{s}$ should be flipped in such a way that the distance from the nearest codewords increases by 1. The search terminates when no further bit flips are possible. The task of searching for the set of codewords nearest to $\mathbf{s}$ and the corresponding Hamming distance is performed by the function NCS (which stands for Nearest Codeword Search Around a Nonzero Vector) in Algorithm \ref{alg:greedy}. A specific implementation of this function for the $(128,64)$ \ac{LDPC} code (or for any rate-$1/2$ code)  is described later (see Algorithm \ref{alg:nearestcw}). Algorithm \ref{alg:greedy} is particularly effective for the $(128,64)$ \ac{LDPC} code or, more in general, for all codes whose weight distribution is non-null for even only (or odd only) values of $w$. In fact, the following results hold.

\begin{algorithm}[tb]
\caption{Stochastic Local Search to Find a Receiver Tail Sequence}
\begin{algorithmic}[1]
\renewcommand{\algorithmicrequire}{\textbf{Input:}}
\renewcommand{\algorithmicensure}{\textbf{Output:}}
        \Require Code generator matrices $\mathbf{G}_{\text{L}}= [\mathbf{I}_{k} \mid \mathbf{P}_{\text{L}}]$ and $\mathbf{G}_{\text{R}}= [\mathbf{P}^{(\text{R})} \mid \mathbf{I}_{k}]$ 
        \Ensure Candidate\_TS

\State $\mathbf{s}\leftarrow\,$all\_zero\_codeword; $t \leftarrow$ 0
        \State Select $i\in[0,n-1]$ uniformly at random and set $s_i \leftarrow  1$
        
            \While{$t < T$} \Comment{$T$ = max \# of iterations}
                \State $(\hat{\mathcal{C}}, d) =$ NCS($\mathbf{G}_{\text{L}}$, $\mathbf{G}_{\text{R}}$, $ \mathbf{s}$) \Comment{see Algorithm \ref{alg:nearestcw}}
\State $\mathcal{I} \leftarrow \emptyset$ \Comment{index set Initialization} 

\State $\textbf{c} \leftarrow \bigwedge_{i \in \hat{\mathcal{C}}} \hat{\textbf{c}}^{(i)}$
\State $\mathcal{I} \leftarrow \mathcal{I} \cup \{i \in [0, n-1] : c_i \wedge s_i = 1 \}$
\State $\textbf{c} \leftarrow \bigwedge_{i \in \hat{\mathcal{C}}} \lnot \, \hat{\textbf{c}}^{(i)}$
\State $\mathcal{I} \leftarrow \mathcal{I} \cup \{i \in [0, n-1] : c_i \wedge \lnot s_i = 1\}$

\If{$\mathcal{I} \neq \emptyset$}
    \State Select $i$ uniformly at random from $\mathcal{I}$
    \State Set $\mathbf{s} \leftarrow \mathbf{s} \oplus
    \mathbf{e}_i$
    \State $t \leftarrow t + 1$
\Else
    \State Candidate\_TS$\,\leftarrow\mathbf{s}$ \Comment{TS at distance $d$}
    \State $t \leftarrow T$ 
\EndIf            
\EndWhile
   
\end{algorithmic}
\label{alg:greedy}
\end{algorithm}

\begin{prop}
    Let us consider $\mathbf{s}\in \mathbb{F}_2^n$, and a binary linear code $\mathcal{C}$ such that $A(w)\neq0$ only for even values of $w$ (and obviously for $w=0$)\footnote{The proof is analogous if the weight distribution is non-null only for odd values of $w$.}. Let us denote as $\hat{\mathcal{C}}$ the set of codewords at minimum Hamming distance $d$ from $\mathbf{s}$. If there exists $\mathbf{e}_j$, for some $j\in[0,n-1]$, such that $||\mathbf{e}_j||=1$ and $\min_{\hat{\mathbf{c}} \in \hat{\mathcal{C}}} \{d_{\mathrm{H}}(\mathbf{s}\oplus\mathbf{e}_j,\hat{\mathbf{c}})\} = d+1$, then $\min_{\mathbf{c} \in \mathcal{C}} \{d_{\mathrm{H}}(\mathbf{s}\oplus\mathbf{e}_j,\mathbf{c})\} \geq d+1$. 
    \label{lem:zets}
\end{prop}
\begin{IEEEproof} 
   By definition, $\mathbf{e}_j$, with $j\in[0,n-1]$, is a vector with unitary Hamming weight, thus it flips one bit of $\mathbf{s}$ when computing $\mathbf{s}\oplus\mathbf{e}_j$. Consequently, $\mathbf{s} \oplus \mathbf{e}_j$ differs from $\mathbf{s}$ at exactly one position. 
   The set $\{\hat{\mathbf{c}}^{(i)}\}$ consists of the codewords at minimum Hamming distance $d$ from $\mathbf{s}$. 
   Furthermore, given the condition that 
   \[
   \min_{0 \leq i \leq z-1} \{d_{\mathrm{H}}(\mathbf{s} + \mathbf{e}_j, \hat{\mathbf{c}}^{(i)})\} = d+1,
   \]
   we know that adding the vector $\mathbf{e}_j$ increases the Hamming distance between $\mathbf{s} + \mathbf{e}_j$ and any codeword in $\{\hat{\mathbf{c}}^{(i)}\}$ by exactly $1$.  
   The set $\{\mathbf{c}^{(i)}\}$ includes the subset $\{\hat{\mathbf{c}}^{(i)}\}$.   Assume, for contradiction, that there exists a codeword $\tilde{\mathbf{c}}^{(i)}$ such that 
   \[
   d_{\mathrm{H}}(\mathbf{s} + \mathbf{e}_j, \tilde{\mathbf{c}}^{(i)}) < d+1.
   \]
   Since $\{\hat{\mathbf{c}}^{(i)}\}$ are the codewords at minimum Hamming distance from $\mathbf{s}$ and adding $\mathbf{e}_j$ increases the distance by exactly $1$ for all $\hat{\mathbf{c}}^{(i)}$, the assumption implies the existence of a codeword  closer to $\mathbf{s} + \mathbf{e}$ than any $\hat{\mathbf{c}}^{(i)}$, contradicting the definition of $\{\hat{\mathbf{c}}^{(i)}\}$. 
   So, the minimum Hamming distance between $\mathbf{s} + \mathbf{e}_j$ and any codeword $\mathbf{c}^{(i)} \in \mathcal{C}$ satisfies 
   \begin{equation}
         \min_{0 \leq i \leq 2^k-1} \{d_{\mathrm{H}}(\mathbf{s} + \mathbf{e}_j, \mathbf{c}^{(i)})\} \geq d+1. 
         \label{eq:dguaraVoro}
   \end{equation}
\end{IEEEproof}

In simple terms, Proposition \ref{lem:zets} states that each $while$ iteration taken by Algorithm \ref{alg:greedy} never decreases the minimum Hamming distance of the tentative solution from the considered set of nearest codewords. In order to find a valid support for $\mathbf{e}_j$, it is possible to exploit the following simple property.

\begin{prop} 
    Let us consider $\mathbf{s}\in \mathbb{F}_2^n$, and a binary linear code. Let us denote the codewords at minimum Hamming distance $d$ from  $\mathbf{s}$ as $\{\hat{\mathbf{c}}^{(0)},\ldots,\hat{\mathbf{c}}^{(z-1)}\}$. If  $\exists i\in[0,n-1] \,|\, \hat{c}^{(j_1)}_i=\hat{c}^{(j_2)}_i=s_i,\forall j_1,j_2 \in[0,z-1]$ then  $\min_{0\leq i\leq z-1} \{d_{\mathrm{H}}(\mathbf{s}',\hat{\mathbf{c}}^{(i)})\} = d+1$, where
    \[
\mathbf{s}'=\mathbf{s}+\mathbf{e}_j
\]
\end{prop}
\begin{IEEEproof}
   We start by observing that $\mathbf{s}$ and $\mathbf{s}'$ differ only in the $i$-th coordinate. Given that $\hat{c}^{(j_1)}_i = \hat{c}^{(j_2)}_i=s_i$ for all $j_1, j_2$ in the indices of the nearest codewords, it follows that the $i$-th component of each $\hat{\mathbf{c}}^{(j)}$ is the same (specifically, it is $s_i$).  Then, flipping the $i$-th bit of $\mathbf{s}$ to create $\mathbf{s}'$ means that $s'_i = s_i \oplus 1 \neq \hat{c}^{(j)}_i$ for all $j$.  Thus, the Hamming distance between $\mathbf{s}'$ and each $\hat{\mathbf{c}}^{(j)}$ increases by exactly $1$ compared to the distance between $\mathbf{s}$ and $\hat{\mathbf{c}}^{(j)}$. Therefore, $d_{\mathrm{H}}(\mathbf{s}', \hat{\mathbf{c}}^{(j)}) = d_{\mathrm{H}}(\mathbf{s}, \hat{\mathbf{c}}^{(j)}) + 1 = d + 1$ for all $j$, and consequently, $\min_{0 \leq j \leq z-1} \{d_{\mathrm{H}}(\mathbf{s}', \hat{\mathbf{c}}^{(j)})\} = d + 1$, as required.
\end{IEEEproof}

From \eqref{eq:dguaraVoro}, we find that also Algorithm \ref{alg:greedy}  returns sequences with ``guaranteed'' minimum Hamming distance from all  codewords.

In Fig. \ref{fig:gen}, we show a bidimensional representation of some Voronoi cells, and trace a possible path from a codeword (starting point) to a candidate termination sequence (red square). Each red bullet represents a tuple in $\mathbb{F}_2^n$. The Hamming distance between each pair of tuples represented by connected red bullets is $1$. 

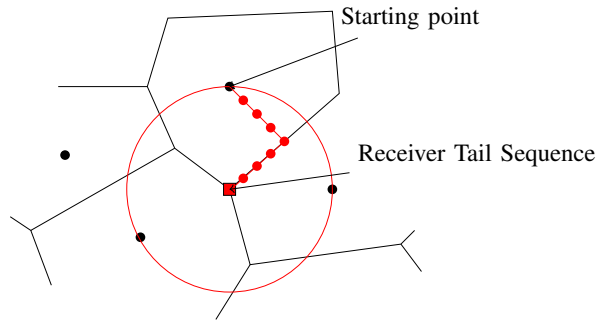
\begin{figure}
\centering
\resizebox{0.95\columnwidth}{!}
{\input{Figs/test_ply}}  
\caption{Simplified representation of the Receiver Tail Sequence search according to Algorithm \ref{alg:greedy}, starting from a codeword, in which each step increases the distance to all nearest codewords, until this is not possible with single steps}
\label{fig:gen}
\end{figure} 

Let us now provide a procedure to find codewords that are closest to a given vector $\mathbf{s}$, for a code having rate $R=1/2$, like the one here of interest. The procedure is divided into \emph{left search} and \emph{right search}. For the left search, it is assumed that the generator is in the following standard form: $\mathbf{G}_{\text{L}} = [\mathbf{I}_{k}\,|\,\mathbf{P}_{\text{L}}]$. The right search is identical, except that one considers $\mathbf{G}_{\text{R}} = [\mathbf{P}_{\text{R}}\,|\,\mathbf{I}_{k}]$.  The whole procedure is described in Algorithm \ref{alg:nearestcw}. This procedure avoids the need for exhaustive codeword enumeration to determine the nearest codeword(s) to a given candidate \ac{TS}. The exhaustive enumeration would be computationally infeasible for the considered code parameters. The following result holds.

\begin{prop}
Let $\mathcal{C}$ be a binary linear code with code rate $\frac{1}{2}$, and let $\mathbf{s}$ be an arbitrary binary input vector. Let $d = \min_{\mathbf{c} \in \mathcal{C}} d_{\mathrm{H}}(\mathbf{s}, \mathbf{c})$ be the minimum Hamming distance from $\mathbf{s}$ to the code $\mathcal{C}$. Then, for large enough $S_{\max}$, Algorithm~\ref{alg:nearestcw} is guaranteed to return $d$.
\end{prop}
\begin{proof} Let $\mathbf{c} \in \mathcal{C}$ be a codeword such that $d_{\mathrm{H}}(\mathbf{s}, \mathbf{c}) = d$. Partition both $\mathbf{s}$ and $\mathbf{c}$ into left and right halves:
\[
\mathbf{s} = (\mathbf{s}_{\mathrm{L}}, \mathbf{s}_{\mathrm{R}}), \quad \mathbf{c} = (\mathbf{c}_{\mathrm{l}}, \mathbf{c}_{\mathrm{r}}).
\]
 Since the total number of differing positions is $d$, it follows from the pigeonhole principle that at least one half differs in at most $\frac{d}{2}$ positions:
\[
\min\{d_{\mathrm{H}}(\mathbf{s}_{\mathrm{L}}, \mathbf{c}_{\mathrm{l}}),\ d_{\mathrm{H}}(\mathbf{s}_{\mathrm{R}}, \mathbf{c}_{\mathrm{r}})\} \leq \frac{d}{2}.
\]
Assume without loss of generality that $d_{\mathrm{H}}(\mathbf{s}_{\mathrm{L}}, \mathbf{c}_{\mathrm{l}}) \leq \frac{d}{2}$. The algorithm enumerates all message vectors whose encoded codewords have left halves within Hamming distance $w\leq \left\lceil \frac{d_s}{2}\right\rceil$ from $\mathbf{s}_{\mathrm{L}}$, where $ d_s\geq d$, so the corresponding message that encodes to $\mathbf{c}$ will be among the candidates. Thus, $\mathbf{c}$ will be evaluated by the algorithm, and since $d_{\mathrm{H}}(\mathbf{s}, \mathbf{c}) = d$ is minimal, the algorithm will identify it as achieving the minimum distance. The hypothesis that $S_{\max}$ is large enough prevents early stopping of the algorithm.
\end{proof}

\begin{algorithm}[t!]
\caption{Nearest Codeword Search Around a Nonzero Vector (NCS)}
\label{alg:nearestcw}
\begin{algorithmic}[1]
\renewcommand{\algorithmicrequire}{\textbf{Input:}}
    \renewcommand{\algorithmicensure}{\textbf{Output:}}
        \Require Code generator matrices $\mathbf{G}_{\text{L}} = [\mathbf{I}_{k} \mid \mathbf{P}_{\text{L}}]$, $\mathbf{G}_{\text{R}} = [\mathbf{I}_{k} \mid \mathbf{P}_{\text{L}}]$, binary vector $\mathbf{s}$ 
        \Ensure Set $\mathcal{S}$ of nearest codewords, distance $d$ to the nearest codewords
    \State $\mathbf{s} = (\mathbf{s}_{\text{L}}, \mathbf{s}_{\text{R}})$ \Comment{Split $\mathbf{s}$ into left and right halves}
    \Statex
    \State $\mathbf{c}_{\text{L}} \gets \mathbf{s}_{\text{L}}\,\mathbf{G}_{\text{L}}$ \Comment{Compute a \emph{reference} codeword}
    \State $\mathbf{u}_{\text{L}} \gets \mathbf{s} \oplus \mathbf{c}_{\text{L}}$ 
    \Statex
    \State $\mathbf{c}_{\text{R}} \gets \mathbf{s}_R\,\mathbf{G}_{\text{R}}$ \Comment{Compute a \emph{reference} codeword}
    \State $\mathbf{u}_{\text{R}} \gets \mathbf{s} \oplus \mathbf{c}_{\text{R}}$ 
    \Statex

    \State $\mathcal{S} \gets \emptyset$ \Comment{Set of nearest codewords}
    \State $d_{s} \gets n + 1$ \Comment{Current minimum distance}
        \Statex
    \State $w \gets 0$ 
\While{$w \leq \lceil d_{s}/2 \rceil$}
    \For{each $k$-bit vector $\mathbf{m}$ with $||\mathbf{m}|| = w$}
        \State $\mathbf{w}_{\text{L}} \gets (\mathbf{m},\,\mathbf{m}\mathbf{P}_{\text{L}})$ \Comment{Left search}
\State $d_{w,{\text{L}}} \gets \left\| \mathbf{w}_{\text{L}} \oplus \mathbf{u}_{\text{L}} \right\|$
\State $\mathbf{w}_{\text{R}} \gets (\mathbf{m}\mathbf{P}_{\text{R}},\,\mathbf{m})$ \Comment{Right search}
\State $d_{w,{\text{R}}} \gets \left\| \mathbf{w}_{\text{R}} \oplus \mathbf{u}_{\text{R}} \right\|$
\State $d_{w} \gets \min(d_{w,\text{L}},\,d_{w,\text{R}})$
\If{$d_w < d_{s}$}
    \State $\mathcal{S} \gets \emptyset$  
    \State $d_s \gets d_w$  \Comment{Update minimum distance}
\ElsIf{$d_w = d_{s}$}
    \If{$|\mathcal{S}| \ge S_{\max}$} 
\Comment{$S_{\max}$ = max \#}
\Statex \Comment{of nearest codewords}
        \State {\bf exit for}
    \EndIf
\EndIf
    \If{$d_{w,\text{L}} = d_w$}
        \State $\mathcal{S} \gets \mathcal{S} \cup \{\mathbf{w}_{\text{L}} \oplus \mathbf{u}_{\text{L}}\}$
    \EndIf
    \If{$d_{w,\text{R}} = d_w$}
        \State $\mathcal{S} \gets \mathcal{S} \cup \{\mathbf{w}_{\text{R}} \oplus \mathbf{u}_{\text{R}}\}$
    \EndIf
    \EndFor
    \State $w \gets w + 1$
\EndWhile
\State $d\gets d_s$
    \State \Return $\mathcal{S}, d$ \Comment{Return both the candidate set and the minimum distance from the nearest codewords}
\end{algorithmic}
\end{algorithm}

Some candidate \acp{DTS} for \acp{CLTU} encoded with the $(128,64)$ \ac{LDPC} code, obtained by running Algorithm \ref{alg:greedy}, are shown in Table \ref{tab:greedy}. Although the local search algorithm is heuristic and therefore its outputs are not provably optimal, we note that it can find sequences with a guaranteed minimum distance that is significantly larger than what is practically achievable using Algorithm \ref{alg:tail_sequence_searchfixed}. This substantial increase in the minimum distance has a significant impact on performance, as demonstrated in Section \ref{sec:numres}. 

\begin{table}[h]
\centering
\caption{Hexadecimal representations of candidate Receiver Tail Sequences obtained with Algorithm \ref{alg:greedy}}
\resizebox{\columnwidth}{!}{ 
\begin{tabular}{|c|c|}
\hline
 $D(\mathbf{v}_G,[0,2^k-1])$ & \textbf{Hexadecimal Representation of $\mathbf{v}$} \\
\hline
 $18$ & FFFFC000\,00000000\,00000000\,00000000 \\
 $19$ & 00008825\,008000A1\,A8402008\,2000C002\\
\hline
\end{tabular}
} 
\label{tab:greedy}
\end{table}

\subsubsection*{Considerations on computational complexity:}

To prevent endless loops caused by the algorithm revisiting previously explored sequences, Algorithm~\ref{alg:greedy} enforces an iteration cap of $T$.   
The computational complexity of each iteration is dominated by the search for the nearest codewords (Step 4). Indeed, apart from the computation of the word ${\bf c}$ (Steps 6 and 8), complexity of all other operations, including the initialization and the bit-flipping procedures, scales at most linearly with the block length~$n$. In general, complexity of the computation of the word $\mathbf{c}$ scales linearly with the size of $\hat{\mathcal{C}}$, which is upper-bounded by the parameter $S_{\max}$ in Algorithm~\ref{alg:nearestcw}, which can be chosen as $o(2^n)$. Note that $|\hat{\mathcal{C}}| = 1$ whenever $d_s \leq \left\lfloor \frac{d_{\min} - 1}{2} \right\rfloor$, in accordance with the error-correcting capability of the code.
So, focusing on the asymptotic behavior, the overall time complexity of  Algorithm \ref{alg:greedy} is  equivalent to that of Algorithm \ref{alg:nearestcw}. 
We thus analyze the computational complexity of Algorithm \ref{alg:nearestcw} (which incorporates ``left'' and ``right'' searches) with respect to $n$. The critical step in terms of time complexity is the enumeration of the sets of $k$-bit vectors $\mathbf{m} \in \mathbb{F}_2^k$ with Hamming weight $\|\mathbf{m}\| = w$ (lines 9-10, Algorithm \ref{alg:nearestcw}). Each of these sets has $\binom{k}{w}$ $k$-bit vectors and, for each of them, the algorithm computes the corresponding codeword suffixes $\mathbf{m} \mathbf{P}_{\text{L}}$ and prefixes $\mathbf{m} \mathbf{P}_{\text{R}}$)  in $O(w k)$ operations (each).
Therefore any main iteration requires 
\begin{equation}\nonumber
O\left(kw\binom{k}{w} \right)
\end{equation}
operations, resulting in a total number of operations bounded by 
\begin{equation}\label{eq:complexity}
O\left(k\sum_{w=0}^{\lceil d/2 \rceil}w\binom{k}{w}\right),
\end{equation}
where $d$ is the minimum Hamming distance from the current $\mathbf{s}$ to the code $\mathcal{C}$.
It is known that, for linear codes with dual distance (i.e., the minimum distance of the dual code) $d_{\perp}$, the overall $D_{\mathcal{C}}$ (defined in \eqref{eq:DC}) is at most $\frac{n}{2}-\left(\frac{1}{2}-o(1)\right)\sqrt{d_\perp n}$ \cite{Bazzi2019}.  
Equation (\ref{eq:complexity}) provides the time-complexity of Algorithm \ref{alg:nearestcw} expressed in terms of the block length $n$ by setting $k = Rn$ with $0 < R < 1$.
This expression highlights that the overall complexity is exponential in the message length $k$. However, Algorithm~\ref{alg:nearestcw} comes with an impressive boost compared to the straightforward nearest codewords search algorithm. The latter requires computing the syndrome of all vectors at Hamming distance up to $d$ from the received vector. Since each syndrome computation involves multiplying the $(n-k) \times n$ parity-check matrix $\mathbf{H}$ by a length-$n$ binary vector of weight $w$, its cost is $O((n-k)w)$. Therefore, the overall complexity of the naive search is 
\[
O\left(n(1-R)\sum_{w=0}^{d} w\binom{n}{w}\right).
\]
For example, to get a distance-$19$ candidate vector through Algorithm \ref{alg:nearestcw} for the considered $(128,64)$ \ac{LDPC} code, around $2\binom{64}{10}\approx 3\cdot 10^{11}$ encoding operations need to be performed. An exhaustive search would take, instead, $\binom{128}{19}\approx 2\cdot 10^{22}$ syndrome calculations.

\subsection{LRT-based detection for the Tail Sequence}\label{subsec:TS4LRT}

Even though both the start sequence and the \ac{TS} are detected by using the same LRT-based optimal approach, their  detection mechanisms differ significantly. According to the standard, the start sequence is a 64-symbol sequence, while the tail sequence consists of 128 symbols. The start sequence is detected using a sliding-window correlator that operates symbol-wise over the received stream, enabling detection at arbitrary positions. In contrast, tail detection can proceed block-wise: since the tail is expected to follow a decoded codeword, it is possible to perform correlation only at predetermined positions. Specifically, the termination detector working points are at offsets of $n$ symbols from the start sequence. 

In light of this, different from conventional start sequence construction, the design of the \ac{DTS} for \ac{LRT}-based detection need not prioritize the optimization of correlation properties or other deterministic properties such as run length, transition density, etc.; instead, it should focus on fulfilling other system-level requirements. For example, since the start sequence for \ac{LDPC}-coded transmissions is $64$ bits long, a good intuition would be that of transforming the last $64$ bits of the $128$-bits long \ac{DTS} into a sequence resembling as much as possible the idle sequence (which is an alternating sequence of $0$s and $1$s). This way, it would be possible to use  \ac{LRT}-based detectors working of segments of the same length for both start and tail sequences.

Any \ac{DTS} can be transformed by exploiting the linearity of the $(128,64)$ code, which implies that the Hamming distance enumerator function is the same taking any of the codewords as reference. Therefore, if we add a codeword to the \ac{DTS}, nothing changes from the decoder perspective. Through binary addition of the last $64$ bits of the \ac{DTS} to a $64$-bit long alternating sequence of $1$s and $0$s (the value of the starting bit can be chosen arbitrarily), a $64$-bit pattern $\mathbf{p}$ is found, which can be unambiguously associated with a codeword $\mathbf{c}^*$. Since the rate of the code is $1/2$, this is always possible. To obtain $\mathbf{c}^*$, it is indeed sufficient to consider $\mathbf{p}$ as an information word, and encode it through binary multiplication by the code generator matrix. 

We also compute the probability of false alarm and the probability of miss detection for the considered sequences using hard correlation, soft correlation, and the \ac{LRT} \cite{Chiani2006}. The probability of false alarm, denoted by \(P_\mathrm{fa}\), is the probability that the detector erroneously detects the presence of the termination sequence when it is absent. Conversely, the probability of miss detection, \(P_\mathrm{md}\), is the probability that the detector fails to detect the presence of the termination when it is actually present. The detection probability is denoted a $P_{\mathrm{d}}$. For hard correlation, the detection metric is given by
\[
\Lambda_\mathrm{H} = \left| \sum_{i=1}^{N} \operatorname{sgn}(r_i) \, v_i \right|,
\]
where \(r_i\) is the \(i\)-th received sample, \(v_i\) is the \(i\)-th element of the termination sequence (with \(v_i \in \{+1,-1\}\)) and $\operatorname{sgn}(\cdot)$ is the sign function. For soft correlation, the metric is instead defined as
\[
\Lambda_\mathrm{S} =  \left| \sum_{i=1}^{N} r_i \, v_i \right|.
\]
The optimal metric in the Neyman-Pearson sense for the \ac{AWGN} channel is given by the \ac{LRT} (often referred to as the Chiani-Martini metric), which is expressed as
\[
\Lambda_\mathrm{CM}= \sum_{i=1}^{N} \ln\left( 1 + e^{-\frac{2}{\sigma^2} \, v_i \, r_i } \right),
\]
where \(\sigma^2\) is the noise variance. In many implementations, the decision rule is then to declare the searched sequence present  if the metric exceeds a threshold, or (for the LRT metric, which decreases when the sequence is present) if it falls below a threshold. By comparing the \(P_\mathrm{fa}\) and \(P_\mathrm{md}\) values corresponding to these metrics, one can assess the detection performance under various noise conditions.

\begin{Exa}
    Consider the distance-$19$ \ac{DTS} shown in the last line of Table \ref{tab:greedy}, that is,
\[
\mathbf{v}= \mathrm{00008825\,008000A1\,A8402008\,2000C002}.
\]
We set the last $64$ symbols to match the idle pattern, through binary addition with the following codeword
\[
\mathbf{c}= \mathrm{6FA5DE77\,A89F2981\,FD15755D\,75559557}.
\]
The final  \ac{DTS} is
\begin{align*}
\mathbf{v}^*_{19} &= \mathbf{v} \oplus \mathbf{c} \\
             &= \mathrm{6FA55652\,A81F2920\,55555555\,55555555}.
\end{align*}

In Fig. \ref{fig:md128} we show the miss detection probability for $\mathbf{v}$ and $\mathbf{v}^*$, respectively. As anticipated, we consider hard correlation, soft correlation, and \ac{LRT}. For each value of the \ac{SNR} per information bit, $E_b/N_0$, the detection threshold is chosen by fixing the probability of false alarm to $10^{-5}$. We observe that the transformation from $\mathbf{v}$ to $\mathbf{v}^*_{19}$ does not affect the overall detection performance, but enables the option of performing detection on $64$ bits instead of $128$ bits, resulting in a slight computational and implementation advantage.

    \begin{figure}
        \centering
        \resizebox{\columnwidth}{!}{\input{Figs/Pmd128TS19pfa_m06}}
     \caption{Comparison of miss detection probabilities using different  detection methods on $\mathbf{v}$ (original) and $\mathbf{v}^*$ (transformed to get second half as idle sequence)}
    \label{fig:md128}
\end{figure}
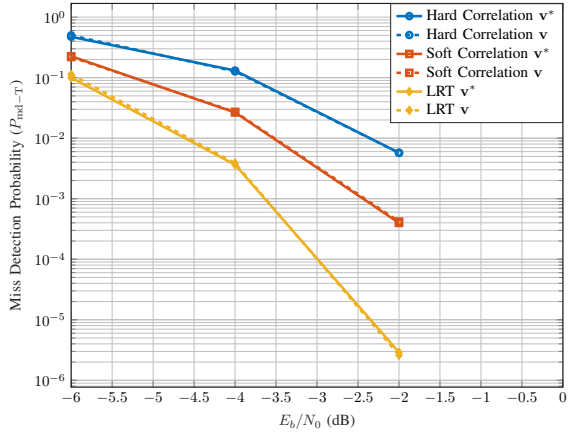

\label{exa:lrt}
\end{Exa}

\section{Performance evaluation}\label{sec:numres}

In this section, using theoretical and numerical arguments, we face the analysis of the most relevant metric for the considered application, that is, the \ac{TC} rejection probability, in the following noted as $P_{\mathrm{TCrej}}$. 

\subsection{Theoretical considerations} 

Given the typical form of a transmitted data stream, that is 
\[
\ldots1010101\,\mathrm{CLTU}_i\,101\ldots101\,\mathrm{CLTU}_{i+1}\,1010101\ldots,
\]
we can determine the reasons for the discard of a \ac{CLTU}, say the $i$-th, at the receiving end:
\begin{itemize}
\item The start sequence detector cannot detect the $i$-th preamble.
\item The start sequence detector generates a false alarm, i.e., it incorrectly detects a preamble where there is none.
\item The decoder does not correctly decode the \ac{CLTU}'s coded data.
\item The termination is not correctly detected.
\end{itemize}
Namely, regarding the last point, assuming that \acp{CLTU} contain the optional \ac{TS} and depending on the termination detection method, we  have that the $i$-th \ac{CLTU} is discarded:
\begin{itemize}
    \item if a decoder-based approach is used for detection,  when the decoder mistakes the $i$-th  \ac{DTS} for a codeword;
    \item if an \ac{LRT}-based (or correlator-based) approach  is used for detection, when:
    \begin{itemize}
        \item the \ac{DTS} detector misses the $i$-th \ac{DTS} AND the decoder mistakes the $i$-th   \ac{DTS} for a codeword;
        \item the \ac{DTS} detector incorrectly signals the presence of the termination sequence when there is none (false alarm).
    \end{itemize}
\end{itemize}

The conditions above can be translated into an explicit analytical formulation. For such a purpose, it is useful to define the following parameters:
\begin{itemize}
    \item $P_{\mathrm{md-S}}$ is the probability of miss detection of the start sequence;
    \item $P_{\mathrm{md-T}}$ is the probability of miss detection of the \ac{TS} (this notation has been already introduced in Fig. \ref{fig:md128}). The corresponding detection probability is denoted as $P_{\mathrm{d-T}}$;
    \item $P_{\mathrm{fa-S}}$ is the probability of false alarm on a sequence formed by at least one symbol of the idle sequence and at least one symbol of the start sequence;
    \item $P_{\mathrm{fa-T}}$ is the probability of false alarm on a codeword;
    \item the \ac{CER} estimates the probability that the decoder fails  to decode a noisy codeword;
    \item $P_{\mathrm{ds-T}}$ is the probability that the decoder erroneously decodes a noisy \ac{DTS} into a valid codeword;
    \item $P_{\mathrm{ds-I}}$ is the probability that the decoder mistakes the first $128$ bits after the last codeword for another codeword, if no termination sequence is used.
\end{itemize}

We observe that all these probabilities can be collected into four groups: failures happening on data (partially) preceding a \ac{CLTU} ($P_{\mathrm{fa-S}}$), failures happening on the start sequence ($P_{\mathrm{md-S}}$), failures happening on coded data ($P_{\mathrm{fa-T}}$ and \ac{CER}), and failures happening on the \ac{DTS} ($P_{\mathrm{md-T}}$ and $P_{\mathrm{ds-T}}$). 

\begin{Rem}
    In the following, we pessimistically assume that all the \acp{CLTU} contain the maximum allowed number of codewords, i.e., $N$. The actual probability that the decoder fails on retrieving a codeword on any \ac{CLTU} is smaller than or equal to $1-(1-\mathrm{CER})^N$.
    \label{rem:upperbounds}
\end{Rem}

Through simple algebra, on the basis of the analysis above, it is easy to verify that, if a decoder based-approach is employed for \ac{CLTU} termination, then

\begin{align}
\nonumber
P_\mathrm{TCrej}^{\mathrm{DB}} &\leq P_{\mathrm{fa-S}} + (1-P_{\mathrm{fa-S}}) \\ \nonumber
&\quad \times \Bigg\{ P_{\mathrm{md-S}} + (1-P_{\mathrm{md-S}}) \Big[ 1-(1-\mathrm{CER})^N \\ \nonumber
&\quad\quad + (1-\mathrm{CER})^N P_{\mathrm{ds-T}} \Big] \Bigg\}.\\ 
\label{eq:tcrejdb}
\end{align}

Instead, assuming that the \ac{LRT}-based detector is employed and that its failure events are independent from those of the decoder (an assumption required because, after detecting the preamble, both the \ac{LRT} detector and the decoder operate on the same received symbols), which is reasonable given that they rely on fundamentally different decision processes, we obtain:
\begin{align}
\nonumber
P_\mathrm{TCrej}^{\mathrm{LRTB}} &\leq P_{\mathrm{fa-S}} \\ \nonumber
&\quad + (1-P_{\mathrm{fa-S}}) \Bigg\{ P_{\mathrm{md-S}}+ (1-P_{\mathrm{md-S}}) \\ \nonumber
&\quad\quad  \times \Big[ 1 - \big((1-P_{\mathrm{fa-T}})(1-\mathrm{CER})\big)^N \\ \nonumber
&\quad\quad\quad + \big((1-P_{\mathrm{fa-T}})(1-\mathrm{CER})\big)^N  \\ \nonumber
&\quad\quad\quad \times P_{\mathrm{ds-T}} P_{\mathrm{md-T}} \Big]\Bigg\}.\\
\label{eq:tcrejlrtb}
\end{align}

To complete the picture and for the sake of comparison, we also explicitly express the \ac{TC} rejection probability for a hypothetical scenario in which \acp{CLTU} do not contain \acp{TS}, and termination is detected directly by the decoder. In this case, we get the following:
\begin{align} 
\nonumber
P_\mathrm{TCrej}^{\mathrm{noTS}} &\leq P_{\mathrm{fa-S}} + (1-P_{\mathrm{fa-S}}) \\ \nonumber
&\quad \times \Bigg\{ P_{\mathrm{md-S}} + (1-P_{\mathrm{md-S}}) \Big[ 1-(1-\mathrm{CER})^N \\ \nonumber
&\quad\quad + (1-\mathrm{CER})^N P_{\mathrm{ds-I}} \Big] \Bigg\}.\\
\label{eq:tcrejnots}
\end{align}
We have run Monte Carlo simulations considering $8$ bits, $16$ bits and $32$ bits long idle sequences. The results, not reported in this paper for the sake of brevity, show a rather weak dependence of the system's performance on the idle sequence length. Therefore, for the sake of ease, we decided to assume that the \ac{CLTU} is followed by a sufficiently long idle sequence, so that the information data are also secure from eventual delays. Under this assumption, the decoder is expected to fail when processing this idle sequence alone.

In \cite{GiulianiAccess2025}, the theoretical probability of missed detection is derived under the assumption that the start sequence can tolerate a certain number of hard bit-flipping errors. In contrast, this paper relaxes that assumption by considering more realistic detection algorithms, which operate on soft information. As a result, all the quantities in the above equations are obtained through Monte Carlo simulations. If only hard bit-flipping errors were present, closed-form expressions for the probabilities of missed detection (for the start and tail sequences) could be used instead.

\subsection{Numerical assessment}

Let us fix a requirement on the \ac{TC} rejection probability $\eta=10^{-5}$. Based on the findings in \cite{GiulianiAccess2025}, which provide an in-depth analysis of the \ac{CCSDS} system performance, we assume that restricting the study of the \ac{SNR} range to between $5.5$ and $6.5$ dB is sufficient to obtain accurate results. Regarding the probabilities of false alarm and miss detection, we can assume that in the considered \ac{SNR} range they are negligible, especially if \ac{LRT} detection is used, as shown in Example \ref{exa:lrt} for termination sequences and in \cite[Fig. 3]{Baldi2016} for the start sequence. In Fig. \ref{fig:combinedROC}, we show the \ac{ROC} curves corresponding to the detection performance, considering different metrics, of the first $64$ bits\footnote{We remind that we designed this \ac{TS} in such a way that the same detector can be used for both the start sequence, which is 64-bits long, and the first half of the \ac{TS}.} of the \ac{TS} with $D(\mathbf{v}_G,[0,2^k-1])=19$. The values of false alarm and detection probability were obtained by simulating the transmission of $10^7$ \ac{BPSK}-modulated random codewords and termination sequences over the \ac{AWGN} channel, respectively. We consider hard correlation, soft correlation, and the \ac{LRT} for $E_b/N_0 \in \{-2, 0, 2\}$~dB. Each point on the resulting ROC curves corresponds to a different detection threshold value. As expected, the \ac{LRT} consistently outperforms soft correlation, which in turn outperforms hard correlation across all considered operating conditions.

\begin{figure*}
    \centering
    \begin{subfigure}[b]{0.48\textwidth}
        \centering
        \resizebox{\textwidth}{!}{\input{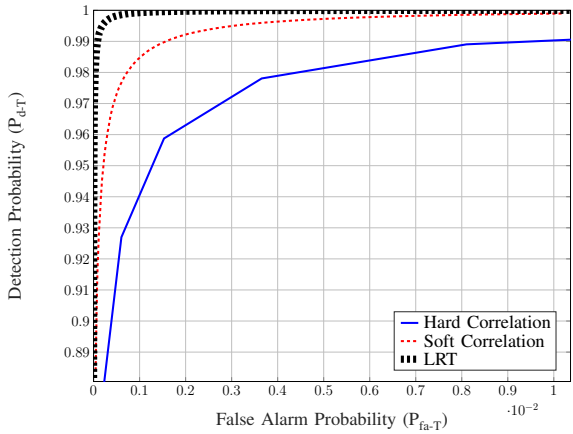}}
        \caption{ROC at $E_b/N_0 = -2$ dB}
        \label{fig:ROC5}
    \end{subfigure}
    \hfill
    \begin{subfigure}[b]{0.48\textwidth}
        \centering
        \resizebox{\textwidth}{!}{\input{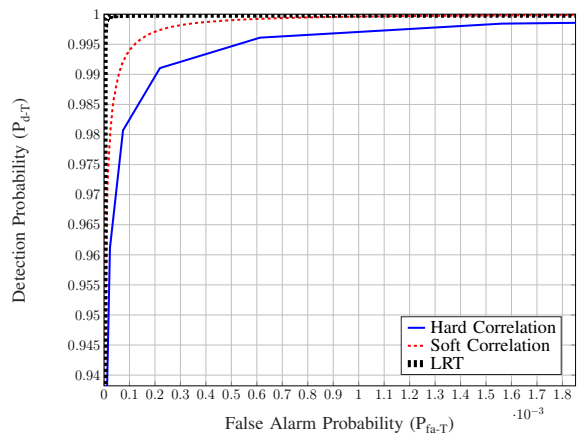}}
        \caption{ROC at $E_b/N_0 = 0$ dB}
        \label{fig:ROC3}
    \end{subfigure}
    
    \begin{subfigure}[b]{0.48\textwidth}
        \centering
        \resizebox{\textwidth}{!}{\input{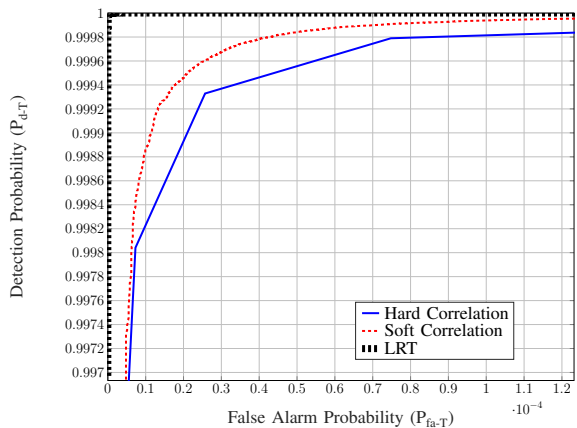}}
        \caption{ROC at $E_b/N_0 = 2$ dB}
        \label{fig:ROC1}
    \end{subfigure}
    
    \caption{\ac{ROC} curves for the termination detection performance on the first $64$ bits of $\mathbf{v}_{19}^*$, considering different metrics and different values of $E_b/N_0$}
    \label{fig:combinedROC}
\end{figure*} 

The \ac{CER} of the $(128,64)$ \ac{LDPC} code has already been assessed under several iterative decoding algorithms in the literature (see, for example, \cite[Figures 11-13]{greenbook}  and \cite[Figure 5]{GiulianiAccess2025}). 
To assess the termination detection performance, the decoder we consider in this paper is the \ac{MSA}, running at most $100$ decoding iterations. We also consider $664$ kB \acp{CLTU}, i.e., $N=40$, when the optional \ac{TS} is included. If not, \acp{CLTU} are $128$ bits shorter.  We run Monte Carlo simulations for the case of \ac{BPSK}-modulated transmissions over the \ac{AWGN} channel, to estimate the decoding success rate when some of the newly proposed \acp{TS} are transmitted ($P_{\mathrm{ds-T}}$), and when an idle sequence of length $128$ is transmitted ($P_{\mathrm{ds-I}}$). The simulation is stopped upon encountering $100$ detection errors. This way, using \eqref{eq:tcrejdb}, \eqref{eq:tcrejlrtb} and \eqref{eq:tcrejnots}, we are able to provide reliable performance in the scenarios of interest: termination detection with a decoder, termination detection with an \ac{LRT} detector/correlator, \ac{TS} not used and (randomized) idle sequence detected with either the decoder or an \ac{LRT}-based detector, respectively. The considered simulation parameters are summarized in Table \ref{tab:simparams}.

\begin{table*}[t!]
\centering
\caption{Main parameters used for numerical experiments}
\label{tab:simparams}
\begin{tabular}{|l|l|}
\hline
\textbf{Parameter} & \textbf{Value} \\
\hline
Target TC rejection probability ($\eta$) & $10^{-5}$ \\
$E_b/N_0$ range & $5.5$ to $6.5$ dB \\
Channel model, modulation scheme & AWGN, BPSK\\
Code & $(128,64)$ LDPC code \cite{bluebook}\\
Decoding algorithm & Min-Sum Algorithm, max 100 iterations \\
Simulation stopping rule & 100 detection errors (per configuration) \\
Detection methods & Hard correlation, soft correlation, LRT \\
Maximum CLTU length & 664 kB with TS; 128 bits shorter without TS \\
\hline
\end{tabular}
\end{table*}

For \ac{CLTU} termination, we consider:
\begin{itemize}
    \item no \ac{TS} is employed. Therefore, we feed the decoder with a  noisy sequence of $128$ alternating $0$s and $1$s (idle sequence), randomized through binary addition with the first $128$ bits generated by the \ac{LFSR}, given in \eqref{eq:lfsr}. In hexadecimal, we have:
    \[
    \mathrm{AA6CCB0F\,3DBC53A0\,39DC7AF4\,640B5D95};
    \]
    \item the \ac{TS} proposed in the standard is employed. In the encapsulation phase, the following sequence is thus included after coded data:
    \[
    \mathrm{55555556\,AAAAAAAA\,55555555\,5555 5555};
    \]
    \item a \ac{TS} for which $D(\mathbf{v}_G,[0,2^k-1])=12$, obtained through Algorithm \ref{alg:tail_sequence_searchfixed}. In the encapsulation phase, the following sequence ($\mathbf{v}_{12}$ in Fig. \ref{fig:combined}) is thus included after coded data:
    \[
     \mathrm{0BF39BF5\,948E906F\,5EC8B2FA\,1314E08C};
    \]
    \item a \ac{TS} for which $D(\mathbf{v}_G,[0,2^k-1])=18$, obtained through Algorithm \ref{alg:greedy}, is employed. In the encapsulation phase, the following sequence ($\mathbf{v}_{18}$ in Fig. \ref{fig:combined}) is thus included after coded data:
    \[
    \mathrm{00C65E5A\,68E906F5\,6C892FA1\,315E08C0};
    \]
    \item a \ac{TS} for which  $D(\mathbf{v}_G,[0,2^k-1])=19$ is employed, corresponding to the design in Section \ref{sec:tsdes}-\ref{subsec:TS4LRT}. In the encapsulation phase, the following sequence ($\mathbf{v}^*_{19}$ in Figs. \ref{fig:combinedROC} and \ref{fig:combined}) is thus included after coded data:
    \[
    \mathrm{909CC808\,C0F62FD5\,39DC7AF4\,640B5D95}.
    \]
\end{itemize}

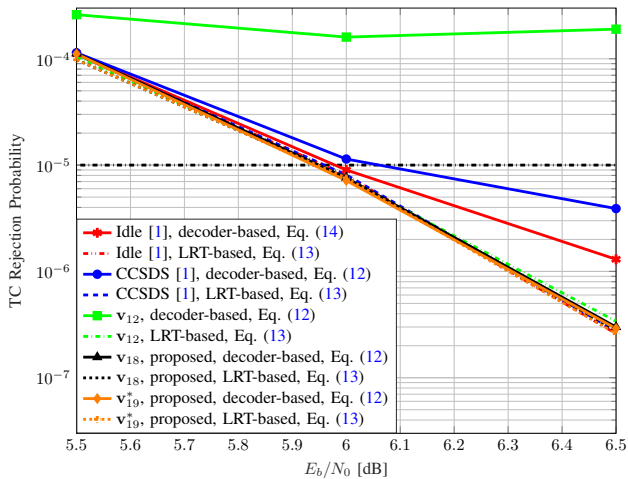
\begin{figure}
        \centering
        \resizebox{\columnwidth}{!}{\input{Figs/rejprob}}
    \caption{Comparison of \ac{TC} rejection probabilities using different termination and detection methods}
    \label{fig:combined}
\end{figure}

\begin{table*}[t!]
\centering
\caption{Termination sequences used for  performance evaluation}
\label{tab:termination_sequences}
\renewcommand{\arraystretch}{1.2}
\begin{tabular}{|c|l|c|p{8.2cm}|}
\hline
\textbf{Label} & \textbf{Description} & \textbf{Hexadecimal Sequence (128 bits) included in encapsulation} \\
\hline
Idle & No TS (LFSR-randomized idle pattern) & \texttt{AA6CCB0F\,3DBC53A0\,39DC7AF4\,640B5D95} \\
CCSDS & Standard TS from \cite{bluebook} &  \texttt{55555556\,AAAAAAAA\,55555555\,55555555} \\
$\mathbf{v}_{12}$ & Proposed TS with Hamming distance 12 & \texttt{0BF39BF5\,948E906F\,5EC8B2FA\,1314E08C} \\
$\mathbf{v}_{18}$ & Proposed TS with Hamming distance 18 & \texttt{00C65E5A\,68E906F5\,6C892FA1\,315E08C0} \\
$\mathbf{v}^*_{19}$ & Proposed TS with Hamming distance 19 & \texttt{909CC808\,C0F62FD5\,39DC7AF4\,640B5D95} \\
\hline
\end{tabular}
\end{table*}

The sequences employed for performance evaluation are reported in Table \ref{tab:termination_sequences}.  From Fig. \ref{fig:combined}, where the threshold $\eta$ is represented by the horizontal black dot-dashed line,  we can draw some very important conclusions:
\begin{itemize}
    \item the performance of the \ac{LRT}-based approach does not depend on the \ac{TS}, in the considered range of $E_b/N_0$; therefore, the colored dotted lines are barely distinguishable;
    \item the system  proposed in \cite{bluebook} has a better performance without \ac{TS} than with the recommended one. Some insights on this issue have been provided in \cite{GiulianiAccess2025};
    \item the fact that the Hamming weight of the \ac{TS} with guaranteed minimum Hamming distance from codewords is relatively low helps the decoder in converging towards the all-zero codeword. Therefore, performance is not as good as the other \acp{TS} with higher values of $D(\mathbf{v})$; 
\item while the performance gain of the optimal \ac{TS} sequence may appear limited for a rejection probability of $10^{-5}$, it becomes more significant at lower rejection rate targets. Since the optimal search is performed off-line and only once, the additional complexity is justified, especially in mission scenarios with stricter reliability requirements;
    \item most important, with a properly designed \ac{TS} with the decoder-based approach, it is possible to obtain the same performance of the \ac{LRT}-based approach. This result represents the main contribution of our work from a practical point of view. In fact, using a decoder-based approach results in savings in terms of computational complexity. In fact, a key distinction arises when comparing \ac{LRT}-based and decoder-based detection methods, namely:
\begin{itemize}
    \item with \ac{LRT}-based \ac{TS} detection, both the start sequence detector and the \ac{TS} detector analyze the received symbol stream, to find preambles and terminations. The decoder remains active only during the presence of a \ac{CLTU};
    \item with decoder-based \ac{TS} detection, the decoder remains active only when processing \acp{CLTU}. However, only a start sequence detector is employed, operating solely during periods without \acp{CLTU}.
\end{itemize}
The above reasoning clearly highlights that a system employing \ac{LRT}-based \ac{TS} detection bears the presence of an additional device, compared to the decoder-based \ac{TS} detection, which inevitably results in additional complexity. 
\end{itemize}

\section{Conclusions and future works}\label{sec:concl}

This work addressed the problem of designing and analyzing \ac{TC} synchronization in space data systems employing short \ac{LDPC} codes. Through theoretical analysis and numerical simulations, we have demonstrated that well-designed \acp{TS} can accommodate different detection strategies, including correlators, \acp{LRT}, and decoder-based approaches. The proposed \ac{TS} design ensures robustness under all these techniques. Furthermore, we have shown that decoder-based detection can achieve \ac{LRT}-equivalent performance, allowing for a simpler, lower-latency, and more efficient on-board implementation. This is a significant advantage, as minimizing receiver complexity is essential for space-based \ac{TC} processing.

As a spark for future works, we conjecture that, if the $E_b/N_0$ working point is high, looking for a word which is distant from all the codewords might not be necessary. In fact, it is well-known that in the high \ac{SNR} region, the performance of iterative decoders suffers from the presence of some harmful objects (i.e., trapping sets, absorbing sets, pseudo-codewords, etc.) \cite{Richardson2003a,Koetter2003,Battaglioni2023Mitch}. These objects are not necessarily distant from codewords, and depend on the considered code and decoder.

Another research line which might affect the \ac{TS} design deals with sequence detection using artificial intelligence. In \cite{deSouza2023,Magrin2019,Khan2024}, the problem of preamble detection in satellite communications is analyzed. Whether these techniques work well for termination detection is still an open question. Joint decoding and detection based on machine learning techniques is proposed in \cite{battaglioni2025mltail}. However, the method has been validated only for a specific value of $E_b/N_0$ and maximum number of decoding iterations, and its general validity is still uncertain.

\section*{Acknowledgment}

The authors heartily thank Dr. Paolo Santini for providing the software related to \cite{stern1989method}.

\bibliographystyle{IEEEtran}
\bibliography{bibliography}

\end{document}

%% file: Figs/block.tex
\begin{tikzpicture}[node distance=2.25cm and 1.75cm, 
    every node/.style={rectangle, draw, text width=4cm, align=center, minimum height=1cm},  auto]

    \node (infoword1) [rectangle, draw=none, fill=none] {Infowords};
    \node (encoder) [below=of infoword1] {Encoder};
    \node (randomizer1) [below=of encoder] {Randomizer};
    \node (cltu) [below=of randomizer1] {Encapsulation of randomized encoded data into a CLTU};
    \node (awgn) [below=of cltu, xshift=3cm] {AWGN}; 

    \node (infoword2) [right=of infoword1, rectangle, draw=none, fill=none] {Decoded infowords};
    \node (decoder) [below=of infoword2] {Decoder};
    \node (randomizer2) [below=of decoder] {De-randomizer};
    \node (detection) [below=of randomizer2] {Detection of start sequence};

    \draw[->] (infoword1) -- (encoder);
    \draw[->] (encoder) -- (randomizer1) node[midway, right, rectangle, draw=none, fill=none,xshift=-0.6cm] {Encoded data};
    \draw[->] (randomizer1) -- (cltu) node[midway, right, rectangle, draw=none, fill=none,xshift=0.3cm] {Randomized encoded data};
    \draw[->] (cltu) |- (awgn);
    \node[below right=1.09cm and -3.25cm of cltu, rectangle, draw=none, fill=none] {CLTU}; 

    \draw[->] (awgn) -| (detection);
    \node[above right=0.2cm and 0.4cm of awgn, rectangle, draw=none, fill=none] {Noisy CLTU}; 

    \draw[->] (detection) -- (randomizer2) node[midway, right, rectangle, draw=none, fill=none] {Noisy randomized encoded data and noisy tail sequence};
    \draw[->] (randomizer2) -- (decoder) node[midway, right, rectangle, draw=none, fill=none] {Noisy encoded data and randomized noisy tail sequence};
    \draw[->] (decoder) -- (infoword2);

\end{tikzpicture}

%% file: Figs/test_ply.tex
\begin{tikzpicture}[
    dot/.style={circle, fill=black, inner sep=0pt, minimum size=4pt}
]
   \draw (0,0) -- (-0.8,0.6) -- (-1.2,1.5) -- (-0.9,2.5) -- (1.5,2.6) -- 
        (1.6,1.4) -- cycle;

    \draw (0,0) -- (0.3,-1.1)--(2.5,-0.8)--(2.7,-0.6);

\draw (2.5,-0.8) -- (2.8,-1.2);

\draw (0.3,-1.1)--(-0.2,-1.9);

\draw (-0.8,0.6) -- (-2.9,-0.6) -- (-2.6,-1.4) ;
\draw (-2.9,-0.6) -- (-3.2,-0.4) ;

\draw (-1.2,1.5) -- (-2.5,1.5);

    \node[dot] at (-2.4,0.5) {};
    \node[dot] at (0,1.5) {}; 
    \node[dot] at (1.5,0) {};
    \node[dot] at (-1.3,-0.7) {};
  \node[draw, fill=red, shape=rectangle, inner sep=0pt, minimum size=5pt] at (0,0) {};

    \draw[red] (0,0) circle (1.5cm);
    
    \draw[red] (0.2,1.3) -- ++(0,0)
          node[dot, red] {} -- ++(0.2,-0.2)
         node[dot, red] {} -- ++(0.2,-0.2)
         node[dot, red] {} -- ++(0.2,-0.2)
         node[dot, red] {} -- ++(-0.2,-0.18)
         node[dot, red] {} -- ++(-0.2,-0.18)
         node[dot, red] {} -- ++(-0.2,-0.18)
         node[dot, red] {} -- ++(-0.2,-0.18)
     ;
    \draw[red](0,1.5)--(0.2,1.3);
    
    \node[right] at (1.5,2.5) {Starting point} edge[->,black] (0,1.5);
     \node[right] at (1.75,0.5) {Receiver Tail Sequence} edge[->,black] (0,0);
    
\end{tikzpicture}

%% file: Figs/Pmd128TS19pfa_m06.tex
%
%
\definecolor{mycolor1}{rgb}{0.00000,0.44700,0.74100}%
\definecolor{mycolor2}{rgb}{0.85000,0.32500,0.09800}%
\definecolor{mycolor3}{rgb}{0.92900,0.69400,0.12500}%
\begin{tikzpicture}

\begin{axis}[%
width=4.521in,
height=3.531in,
at={(0.758in,0.516in)},
scale only axis,
xmin=-6,
xmax=0,
xlabel style={font=\color{white!15!black}},
xlabel={$E_b/N_0 \text{ (dB)}$},
ymode=log,
ymin=e-06,
ymax=0,
yminorticks=true,
ylabel style={font=\color{white!15!black}},
ylabel={$\text{Miss Detection Probability}$ $\text{(}P_{\mathrm{md-T}}\text{)}$},
axis background/.style={fill=white},
xmajorgrids,
ymajorgrids,
yminorgrids,
legend style={at={(1,1)},legend cell align=left, align=left, draw=white!15!black}
]
\addplot [color=mycolor1, line width=0.6mm,mark=o, mark options={solid, mycolor1}]
  table[row sep=crcr]{%
-6 4.74E-01 \\    
-4 1.311E-01\\
-2 5.75E-03\\
};
\addlegendentry{Hard Correlation $\mathbf{v}^*$} %

\addplot [dashed,color=mycolor1, line width=0.6mm,mark=o, mark options={solid, mycolor1}]
  table[row sep=crcr]{%
  -6 5E-01\\
-4 1.27E-01\\
-2 5.72E-03\\
};
\addlegendentry{Hard Correlation $\mathbf{v}$}

\addplot [color=mycolor2,line width=0.6mm, mark=square, mark options={solid, mycolor2}]
  table[row sep=crcr]{%
-6 2.26E-01\\
-4 2.66E-02\\
-2 3.98E-04\\
};
\addlegendentry{Soft Correlation $\mathbf{v}^*$} %

\addplot [dashed , color=mycolor2,line width=0.6mm, mark=square, mark options={solid, mycolor2}]
  table[row sep=crcr]{%
 -6 2.192E-01\\
 -4 2.73E-02\\
-2 4.2E-04\\
};
\addlegendentry{Soft Correlation $\mathbf{v}$}

\addplot [color=mycolor3,line width=0.6mm, mark=diamond, mark options={solid, mycolor3}]
  table[row sep=crcr]{%
-6 9.97E-02\\
-4 3.64E-03\\
-2 2.9E-06\\
};
\addlegendentry{LRT $\mathbf{v}^*$} %

\addplot [dashed,color=mycolor3,line width=0.6mm, mark=diamond, mark options={solid, mycolor3}]
  table[row sep=crcr]{%
 -6 1.13E-01   \\
 -4 3.88E-03\\
 -2 2.6E-06 \\
};
\addlegendentry{LRT $\mathbf{v}$}

\end{axis}

\begin{axis}[%
width=5.833in,
height=4.375in,
at={(0in,0in)},
scale only axis,
xmin=0,
xmax=1,
ymin=0,
ymax=1,
axis line style={draw=none},
ticks=none,
axis x line*=bottom,
axis y line*=left
]
\end{axis}
\end{tikzpicture}%

%% file: Figs/rejprob.tex
\begin{tikzpicture}
        \begin{axis}[%
        width=4.521in,
        height=3.566in,
        scale only axis,
        unbounded coords=jump,
        xmin=5.5,
        xmax=6.5,
        xlabel style={font=\color{white!15!black}},
        xlabel={$E_b/N_0$ [dB]},
        ymode=log,
        ymin=3e-08,
        ymax=3e-04,
        yminorticks=true,
        ylabel style={font=\color{white!15!black}},
        ylabel={TC Rejection Probability},
        axis background/.style={fill=white},
        title style={font=\bfseries},
        xmajorgrids,
        ymajorgrids,
        yminorgrids,
        legend style={at={(-0,0)}, anchor=south west, legend cell align=left, align=left, draw=white!15!black,fill=white, 
    fill opacity=1, 
    draw opacity=1,  
    text opacity=1   
    }
        ]

    \draw[black,  dash dot, line width=0.6mm] (axis cs:1,1e-5) -- (axis cs:7,1e-5);

\addplot [color=red,mark=asterisk, mark size=3pt,line width=0.6mm]
          table[row sep=crcr]{%
          5.5 0.1115E-03\\
          6 0.009E-3\\
          6.5 0.0013E-03\\
        };
        \addlegendentry{Idle \cite{bluebook}, decoder-based, Eq. \eqref{eq:tcrejnots}}

        \addplot [color=red, dash dot dot, line width=0.6mm]
          table[row sep=crcr]{%
          5.5 0.1101E-03\\
          6 0.0080E-03\\
          6.5  0.00026E-03\\
         };
        \addlegendentry{Idle \cite{bluebook}, LRT-based, Eq. \eqref{eq:tcrejlrtb}}

        \addplot [color=blue, mark=*,line width=0.6mm]
          table[row sep=crcr]{%
          5.5 0.1139E-03\\
          6 0.0114E-03\\
          6.5 0.0039E-03\\
        };
        \addlegendentry{CCSDS \cite{bluebook}, decoder-based, Eq. \eqref{eq:tcrejdb}}

        \addplot [color=blue, dashed,line width=0.6mm]
          table[row sep=crcr]{%
           5.5 0.11E-03\\
          6 0.00810E-03\\
          6.5  0.000275E-03\\
        };
        \addlegendentry{CCSDS \cite{bluebook}, LRT-based, Eq. \eqref{eq:tcrejlrtb}}

        \addplot [color=green, mark=square*,line width=0.6mm]
          table[row sep=crcr]{%
          5.5 2.6E-04 \\
          6 1.6E-04\\
          6.5 1.9E-04\\
        };
        \addlegendentry{$\mathbf{v}_{12}$, decoder-based, Eq. \eqref{eq:tcrejdb}  }

        \addplot [color=green, dash dot,line width=0.6mm]
          table[row sep=crcr]{%
          5.5 0.105E-03\\
          6 0.00740E-03\\
          6.5  0.00034E-03\\
        };
        \addlegendentry{$\mathbf{v}_{12}$, LRT-based, Eq. \eqref{eq:tcrejlrtb}}

        \addplot [black,line width=0.6mm,mark=triangle*]
          table[row sep=crcr]{%
          5.5 0.1125E-03\\
          6 0.0076E-03\\
          6.5 0.0003E-03 \\
        };
        \addlegendentry{$\mathbf{v}_{18}$, proposed, decoder-based, Eq. \eqref{eq:tcrejdb}}

        \addplot [color=black, dotted,line width=0.6mm]
          table[row sep=crcr]{%
          5.5 0.098E-03\\
          6 0.00760E-03\\
          6.5  0.00027E-03\\
        };
        \addlegendentry{$\mathbf{v}_{18}$, proposed, LRT-based, Eq. \eqref{eq:tcrejlrtb}}

        \addplot [orange,line width=0.6mm,mark=diamond*]
          table[row sep=crcr]{%
          5.5 0.11125E-03\\
          6 0.00716E-03\\
          6.5 0.00029E-03 \\
        };
        \addlegendentry{$\mathbf{v}^*_{19}$, proposed, decoder-based, Eq. \eqref{eq:tcrejdb}}

\addplot [color=orange, mark=+,dotted,line width=0.6mm]
          table[row sep=crcr]{%
          5.5 0.098E-03\\
          6 0.00760E-03\\
          6.5  0.00026E-03\\
        };
        \addlegendentry{$\mathbf{v}^*_{19}$, proposed, LRT-based, Eq. \eqref{eq:tcrejlrtb}}

        \end{axis}
        \end{tikzpicture}

%% file: main.bbl
\begin{thebibliography}{10}
\providecommand{\url}[1]{#1}
\csname url@samestyle\endcsname
\providecommand{\newblock}{\relax}
\providecommand{\bibinfo}[2]{#2}
\providecommand{\BIBentrySTDinterwordspacing}{\spaceskip=0pt\relax}
\providecommand{\BIBentryALTinterwordstretchfactor}{4}
\providecommand{\BIBentryALTinterwordspacing}{\spaceskip=\fontdimen2\font plus
\BIBentryALTinterwordstretchfactor\fontdimen3\font minus \fontdimen4\font\relax}
\providecommand{\BIBforeignlanguage}[2]{{%
\expandafter\ifx\csname l@#1\endcsname\relax
\typeout{** WARNING: IEEEtran.bst: No hyphenation pattern has been}%
\typeout{** loaded for the language `#1'. Using the pattern for}%
\typeout{** the default language instead.}%
\else
\language=\csname l@#1\endcsname
\fi
#2}}
\providecommand{\BIBdecl}{\relax}
\BIBdecl

\bibitem{bluebook}
{Consultative Commitee for Space Data Systems}, \emph{{TC} {S}ynchronization and {C}hannel {C}oding. {I}ssue. {R}ecommendation for {S}pace {D}ata {S}ystem {S}tandard}, {B}lue Book CCSDS 231.0-B-4 Washington, DC, USA, July 2021.

\bibitem{hocquenghem1959codes}
A.~Hocquenghem, ``Codes correcteurs d'erreurs,'' \emph{Chiffres}, vol.~2, pp. 147--156, 1959.

\bibitem{bose1960class}
R.~Bose and D.~Ray-Chaudhuri, ``On a class of error correcting binary group codes,'' \emph{Information and Control}, vol.~3, no.~1, pp. 68--79, 1960.

\bibitem{ccsds_orange}
{CCSDS}, ``{Research and Development for Space Data System Standards: Short Block Length LDPC Codes for TC Synchronization and Channel Coding Experimental Specification},'' {Consultative Committee for Space Data Systems (CCSDS)}, Orange Book CCSDS 231.1-O-1, Apr. 2015, experimental Specification.

\bibitem{bertinelli_2014}
M.~Bertinelli, M.~Baldi, F.~Chiaraluce, N.~Maturo, E.~Paolini, M.~Chiani, R.~Garello, and P.~Dhakal, ``{CLTU Termination Techniques for TC Links},'' CCSDS Fall Meeting, London, UK, Tech. Rep., Oct. 2014.

\bibitem{bertinelli_2015}
M.~Bertinelli, M.~Baldi, F.~Chiaraluce, N.~Maturo, E.~Paolini, and R.~Garello, ``{CLTU Termination Issues},'' CCSDS Fall Meeting, Darmstadt, Germany, Tech. Rep., Nov. 2015.

\bibitem{slide_ken}
K.~Andrews, ``{Tail Sequence Search for the (128, 64) Uplink LDPC Code},'' CCSDS Spring Meeting, San Antonio, Texas, USA, Tech. Rep., Mar. 2017.

\bibitem{Baldi2016}
M.~Baldi, M.~Bertinelli, F.~Chiaraluce, P.~Closas, R.~Garello, N.~Maturo, M.~Navarro, J.~M. Palomo, E.~Paolini, S.~Pfletschinger, P.~F. Silva, L.~Simone, and J.~Vilà-Valls, ``{NEXCODE: Next generation uplink coding techniques},'' in \emph{Proc. 2016 International Workshop on Tracking, Telemetry and Command Systems for Space Applications (TTC)}, Noordwijk, Netherlands, 2016, pp. 1--8.

\bibitem{massey1972optimum}
J.~L. Massey, ``{Optimum frame synchronization},'' \emph{IEEE Transactions on Communications}, vol.~20, no.~2, pp. 115--119, 1972.

\bibitem{Chiani2006}
M.~Chiani and M.~Martini, ``{On sequential frame synchronization in AWGN channels},'' \emph{IEEE Transactions on Communications}, vol.~54, no.~2, pp. 339--348, 2006.

\bibitem{Pfletschinger2015}
S.~Pfletschinger, M.~Navarro, and P.~Closas, ``{Frame synchronization for next generation uplink coding in deep space communications},'' in \emph{Proc. 2015 IEEE Global Communications Conference (GLOBECOM)}, San Diego, CA, USA, 2015, pp. 1--6.

\bibitem{GiulianiAccess2025}
R.~Giuliani, M.~Battaglioni, M.~Baldi, F.~Chiaraluce, and N.~Maturo, ``{Telecommand rejection probability in CCSDS-compliant LDPC-coded space transmissions with tail sequence},'' \emph{IEEE Access}, vol.~13, pp. 8924--8940, 2025.

\bibitem{Lanctot2003}
K.~Lanctot, M.~Li, B.~Ma, S.~Wang, and L.~Zhang, ``Distinguishing string selection problems,'' \emph{Information and Computation}, vol. 185, no.~1, pp. 41--55, 2003.

\bibitem{Frances1997}
M.~Frances and A.~Litman, ``On covering problems of codes,'' \emph{Theory of Computing Systems}, vol.~30, no.~2, pp. 113--119, 1997.

\bibitem{McLoughlin1984}
A.~McLoughlin, ``The complexity of computing the covering radius of a code,'' \emph{IEEE Transactions on Information Theory}, vol.~30, no.~6, pp. 800--804, 1984.

\bibitem{Guruswami2005}
V.~Guruswami, D.~Micciancio, and O.~Regev, ``The complexity of the covering radius problem,'' \emph{Computational Complexity}, vol.~14, no.~2, pp. 90--121, 2005.

\bibitem{Mazumdar2013}
A.~Mazumdar, Y.~Polyanskiy, and B.~Saha, ``{On Chebyshev radius of a set in Hamming space and the closest string problem},'' in \emph{Proc. 2013 IEEE International Symposium on Information Theory (ISIT)}, Istanbul, Turkey, 2013, pp. 1401--1405.

\bibitem{Li2002}
M.~Li, B.~Ma, and L.~Wang, ``On the closest string and substring problems,'' \emph{Journal of the ACM}, vol.~49, no.~2, p. 157–171, 2002.

\bibitem{Meneses2004}
C.~N. Meneses, Z.~Lu, C.~A.~S. Oliveira, and P.~M. Pardalos, ``Optimal solutions for the closest-string problem via integer programming,'' \emph{INFORMS Journal on Computing}, vol.~16, no.~4, pp. 419--429, 2004.

\bibitem{Dinu2019}
L.~P. Dinu, B.~C. Dumitru, and A.~Popa, ``Algorithms for closest and farthest string problems via rank distance,'' in \emph{Theory and Applications of Models of Computation}, T.~Gopal and J.~Watada, Eds.\hskip 1em plus 0.5em minus 0.4em\relax Cham: Springer International Publishing, 2019, pp. 154--171.

\bibitem{stern1989method}
J.~Stern, ``A method for finding codewords of small weight,'' in \emph{Coding theory and applications}, ser. Lecture Notes in Computer Science, G.~D. Cohen and J.~Wolfmann, Eds.\hskip 1em plus 0.5em minus 0.4em\relax New York: Springer, 1989, vol. 388, pp. 106--113.

\bibitem{Bazzi2019}
L.~Bazzi, ``On the covering radius of small codes versus dual distance,'' \emph{IEEE Transactions on Information Theory}, vol.~65, no.~1, pp. 174--183, 2019.

\bibitem{greenbook}
{Consultative Commitee for Space Data Systems}, \emph{{TC} {S}ynchronization and {C}hannel {C}oding - {S}ummary of {C}oncept and {R}ationale. {I}nformational {R}eport}, {G}reen Book CCSDS 230.1-G-3 Washington, DC, USA, Oct. 2021.

\bibitem{Richardson2003a}
T.~Richardson, ``Error floors of {LDPC} codes,'' in \emph{Proc. 41st Annual Allerton Conf.}, Monticello, IL, USA, Oct. 2003, pp. 1426--1435.

\bibitem{Koetter2003}
R.~Koetter and P.~O. Vontobel, ``Graph covers and iterative decoding of finite-length codes,'' in \emph{Proc. 3rd International Symposium on Turbo Codes and Related Topics}, Brest, France, Sep. 2003, pp. 75--82.

\bibitem{Battaglioni2023Mitch}
M.~{Battaglioni}, F.~{Chiaraluce}, M.~{Baldi}, M.~{Pacenti}, and D.~{Mitchell}, ``Optimizing quasi-cyclic spatially coupled {LDPC} codes by eliminating harmful objects,'' \emph{EURASIP Journal on Wireless Communications and Networking}, no.~67, 2023.

\bibitem{deSouza2023}
J.~H.~I. de~Souza and T.~Abrão, ``Deep learning-based activity detection for grant-free random access,'' \emph{IEEE Systems Journal}, vol.~17, no.~1, pp. 940--951, 2023.

\bibitem{Magrin2019}
D.~Magrin, C.~Pielli, v.~Stefanović, and M.~Zorzi, ``{Enabling LTE RACH collision multiplicity detection via machine learning},'' in \emph{Proc. 2019 International Symposium on Modeling and Optimization in Mobile, Ad Hoc, and Wireless Networks (WiOPT)}, Avignon, France, 2019, pp. 1--8.

\bibitem{Khan2024}
M.~U. Khan, E.~Testi, E.~Paolini, and M.~Chiani, ``Preamble detection in asynchronous random access using deep learning,'' \emph{IEEE Wireless Communications Letters}, vol.~13, no.~2, pp. 279--283, 2024.

\bibitem{battaglioni2025mltail}
M.~Battaglioni, R.~Giuliani, F.~Chiaraluce, and M.~Baldi, ``{Machine learning-based tail sequence detection in LDPC-coded space transmissions},'' in \emph{Proc. IEEE Wireless Communications and Networking Conference (WCNC)}, Milan, Italy, Mar. 2025.

\end{thebibliography}
